\documentclass[twocolumn]{autart}     % Enable this line and disable the 
\usepackage{graphicx}
\usepackage{subcaption}
\usepackage{amsmath}
\usepackage{amssymb}
\usepackage{accents}
\usepackage{cases}
\usepackage{tikz,pgfplots}
\pgfplotsset{compat=newest}
\usepackage{xcolor}
\usepackage{caption}
\usepackage{setspace}
\usepackage{booktabs}
\usepackage{array}
\usepackage{cite}
\usepackage{algorithm}
\usepackage{algorithmic}
\usepackage{natbib}
\usetikzlibrary{patterns}
\usetikzlibrary{graphs}
\usetikzlibrary{graphs.standard}

\newtheorem{proposition}{Proposition}
\newtheorem{theorem}{Theorem}
\newtheorem{definition}{Definition}
\newtheorem{lemma}{Lemma}
\newtheorem{corollary}{Corollary}
\newtheorem{remark}{Remark}
\newtheorem{assumption}{Assumption}

\newtheorem{example}{Example}

\newcommand{\revise}[1]{\textcolor{black}{#1}}
\newcommand{\modify}[1]{\textcolor{black}{#1}}

\DeclareMathOperator*{\1}{\textbf{1}} 
\DeclareMathOperator*{\cen}{Cen}

\DeclareMathOperator*{\Conv}{Conv}

\makeatletter
\pgfmathdeclarefunction{alpha}{1}{%
	\pgfmathint@{#1}%
	\edef\pgfmathresult{\pgffor@alpha{\pgfmathresult}}%
}

\begin{document}

\begin{frontmatter}
	
	\title{Resilient Multi-Dimensional Consensus in Adversarial Environment \thanksref{footnoteinfo}}               % Title, preferably not more 

	\thanks[footnoteinfo]{This work is supported by the National Key Research and Development Program of China under Grant 2018AAA0101601. (Corresponding author: Yilin Mo. )}
	
	\author[THU]{Jiaqi Yan}\ead{jiaqiyan@tsinghua.edu.cn},    % Add the 
	\author[TONGJI]{Xiuxian Li}\ead{xli@tongji.edu.cn},               % e-mail address 
	\author[THU]{Yilin Mo}\ead{ylmo@tsinghua.edu.cn}, 
	\author[NTU]{Changyun Wen}\ead{ecywen@ntu.edu.sg}  % (ead) as shown
	
	\address[THU]{Department of Automation and BNRist, Tsinghua University, Beijing, China}
	\address[TONGJI]{Department of Control Science and Engineering, College of Electronics and Information Engineering, Institute for Advanced Study, and Shanghai Research Institute for Intelligent Autonomous Systems, Tongji University, Shanghai, China}
	\address[NTU]{School of Electrical and Electronic Engineering, Nanyang Technological University, Singapore}  % Please supply                

	                             % keyword list or with the 
	% help of the Automatica 
	% keyword wizard

	\begin{abstract}                          % Abstract of not more than 200 words.
This paper considers the multi-dimensional consensus in networked systems, where some of the agents might be misbehaving (or faulty). Despite the influence of these misbehaviors, the benign agents aim to reach an agreement while avoiding being seriously influenced by the faulty ones. To this end, this paper first considers a general class of consensus algorithms, where each benign agent computes an ``auxiliary point" based on the received values and moves its state toward this point. Concerning this generic form, we present conditions for achieving resilient consensus and obtain a lower bound on the exponential convergence rate. Assuming that the number of malicious agents is upper bounded, two specific resilient consensus algorithms are further developed based on the obtained conditions. Particularly, the first solution, based on Helly's Theorem, achieves the consensus within the convex hull formed by the benign agents' initial states, where the auxiliary point can be efficiently computed through linear programming. On the other hand, the second algorithm serves as a ``built-in'' security guarantee for standard average consensus algorithms, in the sense that its performance coincides exactly with that of the standard ones in the absence of faulty nodes while also resisting the serious influence of the misbehaving ones in adversarial environment. Some numerical examples are provided in the end to verify the theoretical results.
	\end{abstract}
	
\end{frontmatter}

%%%%%%%%%%%%%%%%%%%%%%%%%%%%%%%%%%%%%%%%%%%%%%%%%%%%%%%%%%%%%%%%%%%%%%%
\section{Introduction}
The past decades have witnessed remarkable research interest in networked systems. One of its fundamental focuses would be the consensus problem, which has been widely investigated in various applications including formation control of mobile robots (\cite{ren2007information,raffard2004distributed,olfati2004consensus}), data fusion in sensor network (\cite{olfati2005consensus,moallemi2006consensus}), \textit{etc}. Given a set of autonomous agents (such as vehicles, sensors), it seeks a distributed protocol that the agents can utilize to reach a common decision/agreement on some value related to their initial states. 

In these years, considerable attention has been paid to the development of consensus algorithms (\cite{olfati2007consensus,ren2007information,wei2012distributed}). These protocols normally base on a hypothesis that every computing agent is trustworthy and cooperates to follow the algorithms throughout the execution. Nevertheless, as the scale of the network increases, it becomes difficult to secure every agent. On one hand, autonomous agents will communicate with each other to make control decisions. This opens the system to malicious attackers who might compromise the data on transmission channels. On the other hand, some agents may not be willing to follow the given rules if they weigh their private interests more than the public ones. It is reported that such misbehaving agents can either dictate the final consensus value, or prevent the network from reaching an agreement (\cite{sundaram2018distributed}). 

Given the wide applications of consensus algorithms in safety-critical systems, the problem of resilient consensus has been investigated in the literature over decades, aiming to achieve an agreement among benign agents while avoiding being significantly influenced by the network misbehaviors. \revise{To this end, some works are proposed in order to identify and isolate the illegal components. For example, \cite{pasqualetti2012consensus} characterize the resiliency of linear consensus
networks. Through a system-theoretic approach, the authors find that the network should be sufficiently connected for the faulty agents to be detected and identified. For the purpose of detection and isolation, one way is to construct a model-based
detection system by using observers ( \cite{silvestre2017stochastic,silvestre2013gossip,shames2011distributed,gallo2020distributed,yan2020attack}). Specifically, one could reconstruct the
process states and make a decision on possible misbehaviors based on the residuals generated by the observers. Alternatively, with lower computational cost, \cite{ramos2021discretetime,ramos2020general} propose resilient consensus algorithms which are capable of detecting and isolating the misbehaving agents within polynomial times.
} 

\revise{Instead of detecting and excluding the faulty agents, another line of research focuses on developing resilient consensus algorithms which guarantee an acceptable system performance even in the presence of false data.} Most approaches adopt the idea of simply ignoring the suspicious values. For example, \cite{dolev1986reaching} consider the consensus problem in a complete network, where an approximate agreement is desired in the presence of misbehaving agents. In order to overrule the effects of malicious nodes, a secure updating strategy is proposed. The essential idea is that each normal agent discards the most extreme values in its neighborhood and updates the state based on the remaining ones at any time. Such protocol has then inspired a family of algorithms, namely Mean-Subsequence Reduced (MSR) algorithm (\cite{kieckhafer1994reaching,vaidya2012iterative,de1998multistep}). %Investigations on the relationship of network topology and the maximum number of tolerable faulty agents are also provided in these works.
\cite{leblanc2013resilient} modify MSR and present a Weighted Mean-Subsequence-Reduced (W-MSR) algorithm. Different from that in MSR, a normal agent only removes the extreme values that are strictly larger or smaller than its own. This mechanism results in its own state being kept at each time and turns out to keep more useful information than MSR. Furthermore, instead of the complete graphs, they analyze W-MSR in more general topologies. A novel property named network robustness is introduced therein, which characterizes the resiliency of W-MSR in terms of graph structure. \revise{Following the similar research line, \cite{shang2019resilient} proposes a private-expressed opinion formation model, which facilitates the benign agents to achieve a resilient agreement while protecting their private opinions from being eavesdropped by the outliers. \cite{senejohnny2019resilience} further extend MSR to the context where each agent has independent clocks and possibly makes updates at arbitrary instants. This work particularly sheds light on developing resilient event-triggered and self-triggered control strategies. Later, \cite{dibaji2015consensus} generalize the above results to second-order systems. }

The MSR-based strategies ensure the resilient consensus in uni-dimensional systems where agents' states are assumed to be scalars. Under certain conditions, the final agreement is guaranteed to be within the interval limited by the minimum and maximum values of normal agents' initial states. The implication on scalar variables, however, produces crucial limitations in various practical applications such as the vehicle formation control on a 2D-plane. \revise{Inspired by this fact, recent research attention has been paid to extending the MSR algorithm to the more general $d$-dimensional spaces, where $d\geq 1$ (e.g., \cite{vaidya2013byzantine,vaidya2014iterative,yan2019safe,mendes2013multidimensional,wang2019resilient,shang2020resilient,shang2021median,leblanc2017resilient}). Specifically, various algorithms are developed in literature to facilitate the healthy agents to achieve an agreement within the convex hull formed by
their initial states. To this end, at each iteration, every benign agent seeks a resilient convex combination, referring to a point within the convex hull formed by its benign neighboring states. Existing works find this point through calculating either the Tverberg points (\hspace{1pt}\cite{vaidya2013byzantine,vaidya2014iterative}) or the intersection of multiple convex sets (\hspace{1pt}\cite{yan2019safe,mendes2013multidimensional,wang2019resilient}). While the results therein are elegant, the calculation is rather costly (\hspace{1pt}\cite{agarwal2008algorithms}) and these works, unfortunately, do not provide an efficient way to do it. This leads to a major concern in applying the existing algorithms. Moreover, notice that most literature studies one or more specific algorithms, while a unified resiliency analysis for the general class of consensus algorithms is still absent.}

This work is inspired by the issues above. In what follows, we summarize our main contents and contributions:

1) To achieve greater generality, we first propose a unified framework integrating most commonly used consensus algorithms, where each normal agent computes an ``auxiliary point" based on its received states and moves its state towards this point at every iteration. This paper conducts a resiliency analysis over this large class of distributed protocols. We provide verifiable conditions for the achievement of resilient consensus and obtain the lower bound on its exponential convergence rate. Leveraging such conditions enables us to characterize and design a specific resilient algorithm in the presence of misbehaving agents.

2) Following the proposed generic framework, two algorithms are further developed in this work. To be specific, the first one facilitates the resilient consensus within the convex hull formed by benign agents' initial states. \modify{We also present an explicit approach for computing the auxiliary point using linear programming. Compared with the existing works (\cite{vaidya2014iterative,mendes2013multidimensional,yan2019safe}), the proposed algorithm yields lower computational cost which will not increase with the size of the in-neighborhood of any agent.} The second algorithm, by using a switching rule, complements Algorithm~\ref{alg:resilient} and acts as a ``built-in'' security guarantee for standard average consensus algorithms. Specifically, the performance of Algorithm~\ref{alg:switch} agrees with that of the standard ones in the absence of faulty nodes while providing additional resiliency when faulty ones exist.

3) Since the idea behind canonical consensus serves as a fundamental principle in many distributed coordination settings, the methods provided in this work offer powerful tools in tackling the faulty components for a large number of consensus-based distributed algorithms.

%Thus the ``safe kernel'' technique should be accommodated to increase the resiliency of such problems as well. To see this, the distributed optimization in high dimensional spaces is next investigated in this paper. 
%
%Distributed optimizations aim to optimize some system performance with only local information. The agents thus cooperatively solve an optimization problem through distributed strategies.  So as to increase the system's security level, a resilient subgradient descent algorithm is also developed in this paper, which extends the aforementioned safe kernel approach with the standard subgradient method. The theoretical analysis shows the solution is ensured to be within a subset of the convex hull formed by normal agents' local minimizers. Hence the system avoids to be affected by the misbehaviors too much.

%The main contribution of this paper is finally summarized below.
%\begin{enumerate}
%\item The resilient multi-dimensional consensus is investigated. To limit the influence of network misbehaviors on normal agents, a ``safe kernel'' based protocol is developed. We present sufficient conditions on network robustness, under which the algorithm is proved to be effective. 
%\item We also study the resilient optimization in this work. With the extension to include subgradient terms, the ``safe kernel'' technique could be accommodated in this scenario and secure the system as well. This indicates that it provides a powerful tool in handling misfunctioning components in multi-dimensional spaces. 
%\end{enumerate}

A preliminary version of Algorithm~\ref{alg:resilient} has been reported in our previous work \cite{yan2019safe}, where most of the proofs are missing. \revise{In specific, the previous work requires that each normal agent calculates a safe kernel, which is the intersection of a group of convex sets. After that, it computes the vertices of this kernel for updates. Therefore, the major concern of applying it is the high computational complexity. Moreover, the approach may fail in a complex network where some agents have abundant neighbors. In contrast, Algorithm~\ref{alg:resilient} is of more lightweight since, instead of the whole safe kernel, only $2d$ points in it are going to be calculated. As proved, this step can be achieved through a linear programming with a lower computational cost that is independent of the network complexity. Notice that by doing so, the problem of vertex enumeration, which is also computationally prohibitive (\hspace{1pt}\cite{avis1992pivoting}), is avoided.} 

The rest of this paper is organized as follows. After introducing the preliminaries and notations in Section~\ref{sec:pre},  Section~\ref{sec:form} formulates the problem of interest. To achieve the resilient consensus in multi-dimensional spaces,
a general  ``auxiliary point'' based updating framework is detailed in Section~\ref{sec:general}, which incorporates most of the existing consensus algorithms. Sections~\ref{sec:alg} and \ref{sec:algo2} provide specific resilient strategies, the efficiency of which is also analyzed therein. We finally test the main results through numerical examples in Section~\ref{sec:casestudy} and conclude the paper in Section~\ref{sec:conclusion} .

\section{Preliminaries}\label{sec:pre}
We start by introducing the preliminaries and notations.

\subsection{Graph theory}
Let us consider a \revise{digraph} $\mathcal{G}=\{\mathcal{V},\mathcal{E}\}$, where $\mathcal{V}$ is the set of agents, and $\mathcal{E}\subset \mathcal{V}\times\mathcal{V}$ is the set of edges. The edge $e_{ij}\in \mathcal{E}$ indicates that agent $i$ can directly receive information from agent $j$. The sets of in-neighbors and out-neighbors of agent $i$ are respectively defined as $$\revise{\mathcal{N}_i^+\triangleq\{j\in \mathcal{V}|e_{ij}\in \mathcal{E}\}, \;\mathcal{N}_i^-\triangleq\{j\in \mathcal{V}|e_{ji}\in \mathcal{E}\}.}$$

As one might imagine, there is a close coupling between network topology and the maximum number of tolerable faulty agents. For the resilient algorithms discussed in this paper, we characterize their security and efficiency in terms of network robustness. It was first introduced in \cite{leblanc2013resilient} and is formally defined below:
\begin{definition}($r$-robustness):
	A \revise{digraph} $\mathcal{G}=\{\mathcal{V},\mathcal{E}\}$ is $r$-robust, if for any pair of disjoint and nonempty subsets $\mathcal{V}_1, \mathcal{V}_2\subsetneq\mathcal{V}$, at least one of the following statements hold:
	\begin{enumerate}
		\item There exists an agent in $\mathcal{V}_1$ such that it has at least $r$ in-neighbors outside $\mathcal{V}_1$;
		\item There exists an agent in $\mathcal{V}_2$ such that it has at least $r$ in-neighbors outside $\mathcal{V}_2$.
	\end{enumerate}
\end{definition}
\begin{definition}(($r,s$)-robustness):
	A \revise{digraph} $\mathcal{G}=\{\mathcal{V},\mathcal{E}\}$ is ($r,s$)-robust, if for any pair of disjoint and nonempty subsets $\mathcal{V}_1, \mathcal{V}_2\subsetneq\mathcal{V}$, at least one of the following statements holds:
	\begin{enumerate}
		\item Any agent in $\mathcal{V}_1$ has at least $r$ in-neighbors outside $\mathcal{V}_1$;
		\item Any agent in $\mathcal{V}_2$ has at least $r$ in-neighbors outside $\mathcal{V}_2$; 
		\item There are no less than $s$ agents in $\mathcal{V}_1\cup\mathcal{V}_2$ such that each of them has at least $r$ in-neighbors outside the set it belongs to ($\mathcal{V}_1$ or $\mathcal{V}_2$).
	\end{enumerate}
\end{definition}
Intuitively, network robustness is a connectivity measure for graphs. It claims that for any two disjoint and nonempty subsets of agents, there are ``many'' agents within these sets that have a sufficient number of in-neighbors from outsides. 

\revise{An efficient way for constructing the network robustness is given in \cite{leblanc2013resilient}, where the authors propose how to construct a robust digraph from an existing one. On the other hand, \cite{leblanc2013algorithms} introduce several algorithms, which can be either centralized or distributed, to check the robustness of a given digraph. Alternatively, \cite{usevitch2020determining} also presents a method for determining the robustness of digraph using mixed integer linear programming.
Since the construction and verification of robust networks are not the main focus of this paper, we refer interested readers to these works for more details. }

\subsection{Attack model}\label{sec:attack}
To model the behaviors of misbehaving/faulty agents, let us denote by $\mathcal{F}$ the set of such nodes. Any agent $i\in \mathcal{F}$ could either be the adversarial one with the value being manipulated by the attacker,  or the non-participant agent who does not follow the predefined updating rule. On the other hand, $\mathcal{B}$ is the collection of benign agents who always obey the given law to compute the desired function. It is clear that $\mathcal{B}\cap\mathcal{F}=\varnothing$ and $\mathcal{B}\cup\mathcal{F}=\mathcal{V}$.

This paper considers Byzantine attack model (\cite{leblanc2013resilient}), where no restrictions are imposed on the transmitted data of faulty agent $i\in\mathcal{F}$. Namely, both the adversarial and non-participant agents are allowed to send out arbitrary and even different data to different neighbors. Furthermore, the faulty agents could also collude among themselves to decide on the deceptive values to be communicated. The threat can be characterized by its scope:
\begin{enumerate}
	\item ($F$-total attack model) There are at most $F$ misbehaving agents in the network. That is, $|\mathcal{F}|\leq  F$.
	\item ($F$-local attack model) There are at most $F$ misbehaving agents in the in-neighborhood of any agent. That is, $|\mathcal{F}\cap\mathcal{N}_i^+|\leq F$, for any agent $i\in\mathcal{V}$.
\end{enumerate}

%\todo{One can further prove that the larger the $F$ is, the more resilient the detector will be under attacks, but at the same time, the more performance degradation will occur during normal operation when the attacker is absent. Therefore, there exists a trade-off between security and efficiency, which can be tuned by choosing a suitable parameter $F$. 

%\begin{remark}
%	As we will see, the resilient strategy proposed later would be effective to work against even the worst-case attacks. Consequently, it is equivalent to replace the cardinality requirement $|\mathcal{F}|= F$ and $|\mathcal{F}\cap\mathcal{N}_i|= F$ by $|\mathcal{F}|\leq F$ and $|\mathcal{F}\cap\mathcal{N}_i|\leq F$, respectively. Thereby, we could conclude that $F$-total attack model is a special case of $F$-local one.
%\end{remark}

%Alternatively, we can also classify the faulty agents by their transmission capabilities. Particularly,
%\begin{enumerate}
%	\item A faulty agent is called a \emph{malicious} agent if it can only send identical information to all of its neighbors at any time step;
%	\item A faulty agent is called a \emph{Byzantine} agent if it could send different data to different neighbors at some time.
%\end{enumerate}
%In practice, if the network is realized through broadcast communication, it is more reasonable to assume a malicious agent model. On the other hand, if the agents are adopting a point-to-point communication protocol, a Byzantine assumption is of great possibility.

\revise{
\begin{remark}
The parameter $F$, which is an upper bound of the number of faulty agents, is assumed to be known. It might be determined by the \textit{a priori} knowledge on the vulnerability of each agent or the quality of each sensor. For instance, the nodes with lower security investment are reasonably assigned with high vulnerability (\hspace{0.5pt}\cite{abbas2017improving}). Alternatively, the quantity $F$ may also be viewed as a design parameter, indicating how many bad agents the network can and is willing to tolerate. As will be proved, 
choosing a larger $F$ increases the resiliency of network but also yields higher computation complexity. Moreover, \cite{shang2020resilient} also introduce a max-consensus process to estimate $F$. 
We should note that in the rich literature dealing with malicious sensors/agents, e.g., \cite{abdelhakim2013distributed,vamvoudakis2014detection,leblanc2013resilient,vaidya2013byzantine}, it is also assumed that the number of faulty nodes is known to the system operator.
\end{remark}}

\subsection{Notations}\label{sec:notation}

For any vector  $x\in\mathbb{R}^d$, $x_\ell$ is defined as the $\ell$-th entry of it. 
Given a finite set $\mathcal{X}\subset \mathbb{R}^d$, let us denote
\begin{equation}\label{eqn:def_widehat}
\begin{split}
&\revise{m_\ell(\mathcal{X})}  \triangleq \min_{x\in\mathcal{X}} x_\ell, \\
&\revise{M_\ell(\mathcal{X})}  \triangleq \max_{x\in\mathcal{X}} x_\ell,
\end{split}
\end{equation}
and
\begin{equation}
\begin{split}
&\mathcal{H}_\mathcal{X} \triangleq \prod_{\ell\in\{1,...,d\}} [m_\ell(\mathcal{X}), M_\ell(\mathcal{X})].
\end{split}
\end{equation}
Namely, $m_\ell(\mathcal{X})$ and $M_\ell(\mathcal{X})$ are respectively the minimum and maximum value among the $\ell$-th entries of points in $\mathcal{X}$, and $\mathcal{H}_\mathcal{X}$ is the hypercube defined through Cartesian product. We use $\cen(\mathcal{X})$ to denote the center of $\mathcal{H}_\mathcal{X}$.
 
Given any finite set $\mathcal{X}$, we denote by $\Conv(\mathcal{X})$ the convex hull formed by $\mathcal{X}$, namely the set of all convex combinations of the points in $\mathcal{X}$.

\begin{nomenclature}
	\vspace{-15pt}
	\begin{deflist}[AAAAA] %[AAAA] if you have 4 letters max for example
		\defitem{$x_\ell$}\defterm{the $\ell$-th entry of $x\in\mathbb{R}^d$}
		\vspace{-10pt}
		\defitem{$\Conv(\mathcal{X})$}\defterm{convex hull formed by the point set $\mathcal{X}$}
		\vspace{-10pt}
		\defitem{$m_\ell(\mathcal{X})$}\defterm{minimum value among the $\ell$-th entries of points in $\mathcal{X}$}
		\vspace{-10pt}
		\defitem{$M_\ell(\mathcal{X})$}\defterm{maximum value among the $\ell$-th entries of points in $\mathcal{X}$}
		\vspace{-10pt}
		\defitem{$\mathcal{H}_\mathcal{X}$}\defterm{$d$-dimensional hypercube defined by Cartesian product of $[m_\ell(\mathcal{X}),M_\ell(\mathcal{X})]$, where $\ell\in\{1,\cdots,d\}$}
		\vspace{-10pt}
		\defitem{$\cen(\mathcal{X})$}\defterm{the center of $\mathcal{H}_\mathcal{X}$}
		\vspace{-10pt}
		\defitem{$m^\mathcal B_\ell(k)$}\defterm{minimum value among the $\ell$-th entries of benign agents' states at time $k$}
		\vspace{-10pt}
		\defitem{$M^\mathcal B_\ell(k)$}\defterm{maximum value among the $\ell$-th entries of benign agents' states at time $k$}
		\vspace{-10pt}
		\defitem{$\Delta_\ell(k)$}\defterm{the distance between $m^\mathcal B_\ell(k)$ and $M^\mathcal B_\ell(k)$}
		\vspace{-10pt}
		\defitem{$\Xi(k)$}\defterm{convex hull formed by the states of benign agents at time $k$}
		%			\vspace{-10pt}
		%			\defitem{$\overline{m}^i_\ell(k)$}\defterm{the $(dF+1)$-th smallest values among the $\ell$-th entries of states in agent $i$'s neighborhood}
		%			\vspace{-10pt}
		%			\defitem{$\overline{M}^i_\ell(k)$}\defterm{the $(dF+1)$-th largest values among the $\ell$-th entries of states in agent $i$'s neighborhood}
		\vspace{-10pt}
		\defitem{$\mathbb{B}_\delta$}\defterm{the $d$-dimensional ball centered at origin with radius $\delta$.}
	\end{deflist}
\end{nomenclature}

%%%%%%%%%%%%%%%%%%%%%%%%%%%%%%%%%%%%%%%%%%%%%%%%%%%%%%%%%%%%%%%%%%%%%%%
\section{Problem Formulation}\label{sec:form}
In this paper, we consider a group of $N$ agents cooperating over a \revise{digraph} $\mathcal{G}=\{\mathcal{V},\mathcal{E}\}$. At any time $k\geq 0$, let $x^i(k)\in\mathbb{R}^d$ denote the current state of agent $i$. The agents are said to reach a consensus if there exists a constant $\bar x$ such that $\lim_{k\rightarrow \infty}x^i(k)=\bar x$ holds for every agent $i$. In particular, if $\bar x=1/N\sum_{i=1}^{N}x^i(0)$, the average consensus is achieved. 

As many practical applications fit into the framework of average consensus, much research effort has been devoted to this field. In the distributed framework, each agent updates according to a generic rule as follows: 
\begin{equation}\label{eqn:update_benign}
x^i(k+1) = \alpha^i(k) x^i(k) + (1-\alpha^i(k)) \tilde{x}^i(k)+ \varepsilon^i(k),
\end{equation} 
where $\tilde{x}^i(k)$ is termed as an \textit{``auxiliary point''} throughout this paper, and is calculated by a certain function of agent $i$'s in-neighboring states. Moreover, $\varepsilon^i(k)$ is termed as a `` residual''. With respect to different ways of calculating $\tilde{x}^i(k)$ and $\varepsilon^i(k)$, various algorithms are developed by the research community to solve the consensus problem, see \cite{ren2007information,raffard2004distributed,olfati2004consensus,olfati2005consensus,moallemi2006consensus} for examples.

It is worth noticing that, the standard approaches implicitly assume that all agents are reliable throughout the execution and cooperate to achieve the desired value. However, as the number of local agents increases, certain concerns arise that make this assumption to be violated. As discussed before, the strong dependence of distributed algorithms on the communication infrastructures creates lots of vulnerabilities for cyber attacks, where the transmitted information might be manipulated by external adversaries. Additionally, a ``non-participant'' agent may exist, that deviates from the normal update rule and sends out self-designed information to its neighbors for its own benefits. Clearly, such misbehaviors would degrade the performance of consensus protocols: they can either prevent the benign agents from reaching a consensus or manipulate the final agreement to be false. In fact, as shown in \cite{sundaram2018distributed}, a single ``stubborn'' agent can cause all agents to agree on an arbitrary value by simply keeping this value constant. 

These security concerns lead to the study of resilient consensus protocols. In this paper, we intend to present a secure strategy to achieve the agreement among healthy agents while raising its resiliency so as to avoid being seriously influenced by the network misbehavior. 

For this purpose, at any time $k$, let us denote by $\Omega(k)$ the set of benign agents' states. We further define
\begin{equation}\label{eqn:H_i}
	\Xi(k) \triangleq \Conv(\Omega(k)),
\end{equation}
which is the convex hull formed by the states of benign agents.
This paper aims to achieve the objectives below, regardless of the initial states and even in the adversarial environment:
\begin{definition}\label{def:resilientCon}\rm{\textbf{[Resilient consensus]}}
The benign agents are said to achieve the resilient consensus if the following conditions are satisfied:
\begin{enumerate}
	\item\emph{Agreement:} As $k$ goes to infinity, it holds for any benign agent $i$ that $\lim_{k\to \infty}x^i(k) = \bar{x}$ for some $\bar{x}\in\mathbb{R}^d$; 
	\item\emph{$\delta$-Validity:} At any time, the states of benign agents remain in the set $\Xi(0) +\revise{\mathbb{B}_\delta},$
	where \revise{$\mathbb{B}_\delta$} is a $d$-dimensional ball centered at origin with radius $\delta<\infty$.
\end{enumerate}
\end{definition}

We elucidate these conditions below. First, the states of the benign agents should converge to the same constant value even in the presence of misbehaving ones. In addition, $\delta$-validity condition guarantees that the influence caused by the misbehaviors will be bounded throughout the execution. Hence, the resilient consensus algorithms protect the benign agents from being significantly misled by the faulty ones. \revise{In this condition, $\delta$ is a design parameter that indicates the resiliency level that the system is willing to introduce. More specifically, choosing $\delta=0$ means the system wants to achieve high resiliency by seeking a decision vector exactly from $\Xi(0)$. On the other hand, increasing $\delta$ means the system is willing to tolerate more performance loss when the network misbehaviors exist. As will be seen in Remark~\ref{rmk:delta}, this may provide more freedom in the algorithm design.}

In this paper, besides discussing the general framework \eqref{eqn:update_benign} and obtaining the sufficient conditions for it to achieve resilient consensus, we would also provide two specific algorithms to enforce such conditions.

%%%%%%%%%%%%%%%%%%%%%%%%%%%%%%%%%%%%%%%%%%%%%%%%%%%%%%%%%%%%%%%%%%%%%%%
\section{Sufficient Conditions on Achieving Resilient Consensus}\label{sec:general}
Clearly, \eqref{eqn:update_benign} represents a general form of distributed consensus algorithms. Therefore, this section will leverage this form and establish sufficient conditions for this framework to facilitate the resilient consensus. These conditions are helpful for us to characterize and design specific resilient algorithms in the presence of misbehaving agents.

At any $\ell\in \mathcal{D}\triangleq \{1,2,\cdots,d\}$, let us respectively denote by $m^\mathcal B_\ell(k)$ and $M^\mathcal B_\ell(k)$ the minimum and maximum value among the $\ell$-th entries of benign agents' states at time $k$. That is,
\begin{equation}\label{eqn:M(k),m(k)}
\begin{split}
&m^\mathcal B_\ell(k) \triangleq m_\ell(\Omega(k)) =\min_{i\in\mathcal{B}} x^i_\ell(k),\\
&M^\mathcal B_\ell(k) \triangleq M_\ell(\Omega(k)) =\max_{i\in\mathcal{B}} x^i_\ell(k),
\end{split}
\end{equation}
where $m_\ell(\Omega(k))$ and $M_\ell(\Omega(k))$ are defined in \eqref{eqn:def_widehat}. The temporal difference is thereby defined as 
\begin{equation}\label{eqn:difference}
\Delta_\ell(k)=M^\mathcal B_\ell(k)-m^\mathcal B_\ell(k).
\end{equation}
%Moreover, the following definitions are imposed at any $\bar{k}\geq k$, and any $\epsilon\in\mathbb{R}$:
%\begin{equation}\label{eqn:V^MandV^m}
%\begin{split}
%\mathcal{V}_1^M(k,\bar{k},\epsilon)\triangleq\{i\in\mathcal{V}: x^i_1(\bar{k})>M_1(k)-\epsilon\},\\
%\mathcal{V}_1^m(k,\bar{k},\epsilon)\triangleq\{i\in\mathcal{V}: x^i_1(\bar{k})<m_1(k)+\epsilon\}.
%\end{split}
%\end{equation}
%Note $\mathcal{V}_1^M(k,\bar{k},\epsilon)$ [resp. $\mathcal{V}_1^m(k,\bar{k},\epsilon)$] includes all agents whose state's first component is greater [resp. less] than $M_1(k)-\epsilon$ [resp. $m_1(k)+\epsilon$] at time $\bar{k}$. We then define 
%\begin{equation}
%\begin{split}
%\mathcal{B}_1^M(k,\bar{k},\epsilon)\triangleq\mathcal{V}_1^M(k,\bar{k},\epsilon)\cap\mathcal{B},\\
%\mathcal{B}_1^m(k,\bar{k},\epsilon)\triangleq\mathcal{V}_1^m(k,\bar{k},\epsilon)\cap\mathcal{B},
%\end{split}
%\end{equation}
%which consist of benign agents in $\mathcal{V}^M(k,\bar{k},\epsilon)$ and $\mathcal{V}^m(k,\bar{k},\epsilon)$, respectively.
For convenience, the following definitions are made for any $\epsilon\in\mathbb{R}$:
\begin{equation}\label{eqn:V^MandV^m}
\begin{split}
\mathcal{V}_\ell^m(k,\epsilon)\triangleq\{i\in\mathcal{V}: x^i_\ell(k)<\epsilon\},\\
\mathcal{V}_\ell^M(k,\epsilon)\triangleq\{i\in\mathcal{V}: x^i_\ell(k)>\epsilon\}.
\end{split}
\end{equation}
From the definitions above, $\mathcal{V}_\ell^m(k,\epsilon)$ [resp. $\mathcal{V}_\ell^M(k,\epsilon)$] includes all agents, the $\ell$-th component of whose state is lower [resp. greater] than $\epsilon$ at time $k$. We then define 
\begin{equation}\label{eqn:B^MandB^m}
\begin{split}
\mathcal{B}_\ell^m(k,\epsilon)&\triangleq\mathcal{V}_\ell^m(k,\epsilon)\cap\mathcal{B},\\
\mathcal{B}_\ell^M(k,\epsilon)&\triangleq\mathcal{V}_\ell^M(k,\epsilon)\cap\mathcal{B},
\end{split}
\end{equation}
which consist of benign agents in $\mathcal{V}_\ell^m(k,\epsilon)$ and $\mathcal{V}_\ell^M(k,\epsilon)$, respectively.

With the above preparations, we claim that, the below conditions are sufficient for \eqref{eqn:update_benign} to guarantee the resilient consensus, the proof of which is given in Appendix~\ref{app:A}:
\begin{theorem}\label{thm:resilientCon}
	Consider the network $\mathcal{G}=(\mathcal{V}, \mathcal{E})$. Suppose that each benign agent $i\in\mathcal{B}$ updates with rule \eqref{eqn:update_benign}, and the following conditions hold for any $k\in\mathbb{Z}_{\geq 0}$ and $\ell \in \mathcal{D}$:
\begin{enumerate}
	\item[\rm{C1}:]{\rm (Weight)} There exists $0<\alpha<1$ such that $\alpha^i(k) \geq \alpha$ and  $1-\alpha^i(k) \geq \alpha$;
	\item[\rm C2:]{\rm (Residual)} The sequence $\{\varepsilon^i(k)\}$ satisfies that
	\begin{equation}\label{eqn:varepsilon}
	\begin{split}
	\sum_{k=0}^{\infty}||\varepsilon^i(k)||<\infty;
	\end{split}
	\end{equation} 
	\item[\rm C3:]{\rm (Auxiliary point)} There exists a finite set $\Lambda^i(k)\subset \Xi(k)$, such that the auxiliary point is chosen from the center of $\mathcal{H}_{\Lambda^i(k)}$. Namely, $\tilde{x}^i(k)=\cen(\Lambda^i(k))$;
	\item[\rm C4:]{\rm (Topology)}  Given any $\epsilon_1, \epsilon_2\in\mathbb{R}$, if $\mathcal{B}_\ell^M(k,\epsilon_1)$ and $\mathcal{B}_\ell^m(k,\epsilon_2)$ are disjoint and nonempty, there exists a benign agent, labeled as $j$, such that at least one of the following statements holds:
	\begin{enumerate}
	\item[\rm 4a)]  $j\in \mathcal{B}_\ell^M(k,\epsilon_1)$ and there exists a point in $\Lambda^j(k)$, the $\ell$-th entry of which is upper bounded by $\epsilon_1$;
	\item[\rm 4b)] $j\in \mathcal{B}_\ell^m(k,\epsilon_2)$ and there exists a point in $\Lambda^j(k)$, the $\ell$-th entry of which is lower bounded by $\epsilon_2$.
	\end{enumerate}
\end{enumerate}
Then the following statements hold, regardless of
the actions of misbehaving agents:

(1) Benign nodes exponentially reach resilient consensus;

(2) It follows for any $k\in\mathbb{Z}_{\geq 0}$ and any $\ell\in\mathcal{D}$ that:
	\begin{equation}\label{eqn:delta}
		\Delta_\ell(k+|\mathcal{B}|)\leq \bigg(1-\frac{(0.5\alpha)^{|\mathcal{B}|}}{2}\bigg)\Delta_\ell(k)+2\sum_{\tau=k}^{k+|\mathcal{B}|-1}\delta(\tau), 
	\end{equation}
	where 
	\begin{equation}\label{eqn:deltadef}
		\delta(k) \triangleq \max_{i\in\mathcal{B}}||\varepsilon^i(k)||.
	\end{equation}

\end{theorem}

%\begin{corollary}\label{lmm:gap}
%	Under the conditions of Theorem \ref{thm:resilientCon}, it holds at any $k\in\mathbb{Z}_{\geq 0}$ and $\ell \in \mathcal{D}$ that:
%	\begin{equation}\label{eqn:delta}
%		\Delta_\ell(k+|\mathcal{B}|)\leq \bigg(1-\frac{(0.5\alpha)^{|\mathcal{B}|}}{2}\bigg)\Delta_\ell(k)+2\sum_{\tau=k}^{k+|\mathcal{B}|-1}\delta(\tau), 
%	\end{equation}
%\end{corollary}

We present a few remarks regarding the listed conditions. First, condition C3 prevents the misbehaving ones from taking arbitrary control over the dynamics of benign ones. To see this, we note that the set $\Lambda^i(k)$ is ``safe", since it is contained in $\Xi(k)$, which is formed by only benign states. Therefore, $\tilde{x}^i(k)$ is safely chosen. C1 and C2 are imposed for the $\delta$-validity condition, where the former condition guarantees that any benign agent always uses the convex combination of some safe points in updating, and the latter ensures the bias to such a safe combination is accumulatively bounded. \revise{Moreover, the lower bound on update weights, i.e., $\alpha$ in C1, guarantees that any benign agent $i$ could be sufficiently influenced by the ``safe state" $\tilde{x}^i(k)$. Notice that similar conditions are widely adopted in both standard consensus algorithms (\hspace{1pt} \cite{nedic2010constrained,matveev2021diffusion,lorenz2010conditions}) and resilient consensus algorithms (\hspace{1pt}\cite{leblanc2013resilient,sundaram2018distributed,yan2019safe}).}
\revise{Finally, C4 plays a crucial role in guaranteeing the agreement condition. Intuitively, it claims that, if the agreement among benign states has not been reached at the $\ell$-th entry, i.e., there exist some benign states, the $\ell$-th entry of which is smaller [resp.  larger] than others, then one of them has some benign in-neighbors with larger [resp.  smaller] $\ell$-th entries so that it can move towards the consensus. Therefore, C4 is associated with the network topology. In Section~\ref{sec:alg}, we would present topological conditions to enforce this statement. }

As a result of Theorem~\ref{thm:resilientCon}, the resilient algorithm protects the states of benign agents from being driven to arbitrary values. Hence, the network could withstand the compromise of partial agents. Furthermore, since the convergence holds regardless of the actions of misbehaving agents, it works effectively even in the worst-case scenario, where the faulty agents could be Byzantine that send arbitrary and different information to different out-neighbors, and could also have full knowledge of graph topology, updating rules, \textit{etc.}

Notice that Theorem~\ref{thm:resilientCon} provides verifiable conditions for checking the resiliency of any consensus algorithm which is in the form \eqref{eqn:update_benign}. We thus can leverage such results to design specific resilient algorithms in adversarial environment.

%\begin{theorem}[\hspace{1pt}{\cite[Theorem 2.7]{ludwig1963helly}}]\label{thm:Minkowski}
%Let $\Lambda\subset \mathbb{R}^d$ be any convex body. It holds that
%\begin{equation}
%\min _{x \in K} \max _{\ell \ni x} \frac{g(x, \ell)}{|K \cap \ell|} \leq \frac{d}{d+1}.
%\end{equation}
%\end{theorem}

%In view of Theorem \ref{thm:Minkowski}, the best ratio is $\rho^* = \frac{d}{d+1}$.

%%%%%%%%%%%%%%%%%%%%%%%%%%%%%%%%%%%%%%%%%%%%%%%%%%%%%%%%%%%%%%%%%%%%%%%
\section{Resilient Consensus Algorithm 1}\label{sec:alg}
In this section, we shall provide a specific algorithm to solve the problem of resilient consensus under the attack model discussed in Section~\ref{sec:attack}. By ensuring the sufficient conditions in Theorem \ref{thm:resilientCon}, we show that the proposed algorithm guarantees the agreement exactly within the convex hull formed by the initial states of benign agents, namely, achieves the resilient consensus with $\delta=0$.

Throughout this section, we impose the following assumption on network topology:
\begin{assumption}\label{assump:connect}
	For any $i\in\mathcal{B}$, it holds that $|\mathcal N^+_i|\geq (d+1)F+1$.
\end{assumption}
To simplify notations, the following definitions are given beforehand:
\begin{definition}\label{def: S}
	Consider a set\footnote{To be more precise, $\mathcal A$ should be defined as a multi-set since we allow duplicate elements in the set, e.g., the states of $m$ agents shall be counted as $m$ points even if some of them may be identical.} $\mathcal{A}\subset \mathbb R^d$ with cardinality $m$. For some $n\in\mathbb{Z}_{\geq 0}$ and $n \leq m$, let $\mathcal{S}(\mathcal{A},n)$ be the set of all its subsets with cardinality $m-n$. 
\end{definition}

It is clear that the set $\mathcal{S}(\mathcal{A},n)$ contains $\binom{m}{n}$ elements, and each of them is associated with a convex hull. The intersection of all these convex hulls plays a crucial role in our algorithm, which is formally defined below: 
\begin{definition}\label{def: intersection}
	Consider a set $\mathcal{A}\subset\mathbb R^d$ with cardinality $m$. For some $n\in\mathbb{Z}_{\geq 0}$ and $n \leq m$, we define $\varPsi (\mathcal{A}, n)$ as 
	\begin{align}\label{eqn:varPsi}
	\varPsi(\mathcal A,n) \triangleq \bigcap_{S\in \mathcal S(\mathcal A,n)} \Conv(S).
	\end{align}
	%the convex hull of each of them is a set of all convex combinations of its points. 
\end{definition}
In view of Definition \ref{def: intersection}, $\varPsi(\mathcal A,n)$ is a subset of the convex hull formed by any $m-n$ points in $\mathcal A$.

Now we are ready to present the resilient algorithm, where each normal agent $i\in \mathcal{B}$ starts with an initial state $x^i(0)\in \mathbb{R}^d$. At any instant $k> 0$, it makes updates as outlined in Algorithm \ref{alg:resilient}.
\begin{algorithm}[h!]
	1:\: Receive the states from all in-neighbors $j\in\mathcal{N}_i^+$, and collect these values in $\mathcal{X}^i(k)$.\\
	2:\: \textbf{for} $p\in\{1,2,...,d\}$ \textbf{do}
		\qquad \begin{enumerate}
		\item[] a): Sort in an ascending order the points in $ \mathcal{X}^i(k)$ based on their $p$-th entries.
		\item[] b): According to the sorted points, pick up the first $(d+1)F+1$ ones and collect them in $\mathcal{Y}^i(p,k)$. Calculate \textit{any} point $y^i(p,k)$, such that $y^i(p,k) \in \varPsi (\mathcal{Y}^i(p,k), F)$.
		\item[] c): Similarly, pick up the last $(d+1)F+1$ ones from the sorted points and collect them in $\mathcal{Z}^i(p,k)$. Compute any point $z^i(p,k) \in \varPsi (\mathcal{Z}^i(p,k), F)$.
	\end{enumerate}
	\quad \textbf{end for}\\
	3:\: Define a point set as
	\begin{equation}\label{eqn:Lambda^i}
	\Lambda^i(k) \triangleq \{y^i(p,k), 1\leq p\leq d\}\cup \{z^i(p,k), 1\leq p\leq d\}.
	\end{equation}
	Agent $i$ chooses the auxiliary point from the center of $\mathcal{H}_{\Lambda^i(k)}$, namely,
		\begin{align}\label{eqn:auxiliary1}
	\tilde{x}^i(k) = \cen(\Lambda^i(k)).
	\end{align}
	It then updates the local state as:
	\begin{equation}\label{eqn:updatewithattack}
	\begin{split}
	x^i(k+1) = \frac{x^i(k)+\tilde{x}^i(k)}{2}.
	\end{split}
	\end{equation} 

	4:\: Transmit the new state $x^i(k+1)$ to out-neighbors $j\in\mathcal{N}_i^-$. 
	\caption{Resilient consensus algorithm $1$ for benign agent $i$}
	\label{alg:resilient}
\end{algorithm}

\revise{The information flow of Algorithm~\ref{alg:resilient} is presented in Fig.~\ref{fig:diag}.} At each time $k$, the normal agent
sorts the received states from every dimension. We would prove later that both $y^i(p,k)$ and $z^i(p,k)$ are ``safe'', as they belong to the convex hull formed by benign states. Calculating either of them requires a subset of received states, which contains exactly $(d+1)F+1$ points in $\mathcal{X}^i(k)$. By collecting $\{y^i(p,k)\}$ and $\{z^i(p,k)\}$ in $\Lambda^i(k)$, the benign agent chooses the auxiliary point from the center of $\mathcal{H}_{\Lambda^i(k)}$ and makes updates. Finally, as every fault-free agent is only required to access the information in its in-neighborhood, Algorithm \ref{alg:resilient} can be implemented in a fully decentralized manner.

\begin{figure*}[!htbp]
	\centering
	\includegraphics[width=0.85\textwidth]{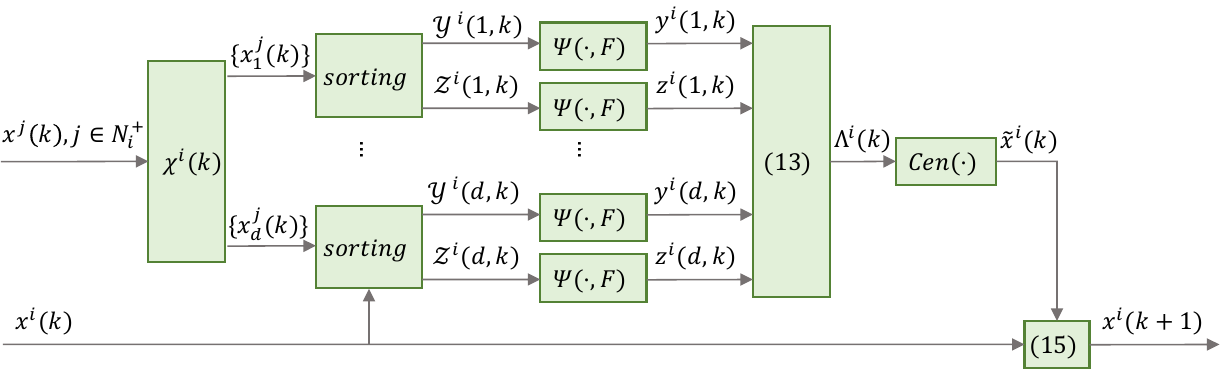}
	\caption{\revise{The schematic diagram of the information flow for Algorithm~\ref{alg:resilient}.}}
	\label{fig:diag}
\end{figure*}

\subsection{Computation of $y^i(p,k)$ and $z^i(p,k)$}\label{sec:middle_points}
Before showing the resiliency of Algorithm~\ref{alg:resilient}, we prove the existence of $y^i(p,k)$ and $z^i(p,k)$ for any $p\in\mathcal{D}$. To this end, it is helpful to introduce Helly's Theorem, which is a key supporting technique of this section:\\

\noindent\textbf{Helly's Theorem (\cite{Danzer1963Helly}). }
\textit{Let $X_1,\cdots, X_q$ be a finite collection of convex subsets in $\mathbb{R}^d$, with $q> d$. If the intersection of every $d + 1$ of these sets is nonempty, then the whole collection has a nonempty intersection. That is,
	$$\bigcap _{j=1}^{q}X_{j}\neq \varnothing. $$}

Below is an immediate result of Helly's Theorem:
\begin{corollary}\label{lm:existence}
Consider a set $\mathcal{A}\subset\mathbb R^d$ with cardinality $m$. For any  $n\in\mathbb{Z}_{\geq 0}$, if $m\geq n(d+1)+1$, then it holds that $$\varPsi (\mathcal{A}, n) \neq \varnothing.$$
\end{corollary}

Combining the results above with Assumption \ref{lm:existence}, it is concluded that both $\varPsi (\mathcal{Y}^i(p,k), F)$ and $\varPsi (\mathcal{Z}^i(p,k), F)$ are nonempty for any $p\in\mathcal{D}$. Therefore, $y^i(p,k)$ and $z^i(p, k)$ are well-defined.

We next discuss the computation of $y^i(p,k)$ in Step 2b), by which $z^i(p,k)$ can be calculated similarly. For simplicity, we omit the dimension index $p$ and time index $k$ in the sequel of this subsection.

Let us denote $\kappa=(d+1)F+1$. Note that $\varPsi (\mathcal{Y}^i, F)$ is the intersection of $r= \binom{\kappa}{F}$ convex hulls, with each of which formed by a set of $q=dF+1$ points. For any of these sets, let us define the matrix with the points in it as
\[
Y_{j}\triangleq\left[\begin{array}{llll}
x^{j_{1}} & x^{j_{2}} & \cdots & x^{j_{q}}
\end{array}\right] \in \mathbb{R}^{d \times q}.
\]	
We then denote
\[
Y\triangleq\operatorname{diag}\left\{Y_{j}, j=1,2, \ldots, r\right\} \in \mathbb{R}^{d r \times q r}.
\]

\begin{example}
	Suppose that $\mathcal{Y}^i=\{x^1,x^2,x^3,x^4\}$ and $F=1, d=2$. One thus has
	\[
	\begin{aligned}
	Y_{1} &=\left[\begin{array}{lll}
	x^{1} & x^{2} & x^{3}
	\end{array}\right], Y_{2}=\left[\begin{array}{lll}
	x^{1} & x^{2} & x^{4}
	\end{array}\right], \\ Y_{3}&=\left[\begin{array}{lll}
	x^{1} & x^{3} & x^{4}
	\end{array}\right], Y_{4}=\left[\begin{array}{lll}
	x^{2} & x^{3} & x^{4}
\end{array}\right], \end{aligned}
	\] and
	$$Y=\operatorname{diag}\left\{Y_1, Y_2, Y_3, Y_4\right\}.$$
\end{example}

Instead of \eqref{eqn:varPsi}, the following lemma provides another way to represent $\varPsi (\mathcal{Y}^i, F)$ in terms of linear constraints:
\begin{lemma}[\cite{wang2019resilient}]\label{lmm:equivalent_y}
	Let \(C \in \mathbb{R}^{r \times r}\) be the circulant matrix with the first row as \(\left[\begin{array}{lllll}1 & -1 & 0 & \cdots & 0\end{array}\right] .\) Then
	\[
	\varPsi (\mathcal{Y}^i, F)=\left\{\frac{1}{r}\left(\mathbf{1}_{r}^{\prime} \otimes I_{d}\right) Y \gamma\right\},
	\]
	for all \(\gamma \in \mathbb{R}^{qr}\) such that
	\begin{equation}\label{eqn:constraint}
	\begin{aligned}
	\begin{bmatrix}
	\left(C \otimes I_{d}\right) Y \\
	\left(I_{r} \otimes \mathbf{1}_{q}^{\prime}\right)
	\end{bmatrix} \gamma &=\begin{bmatrix}
	\mathbf{0}_{dr} \\
	\mathbf{1}_{r}
	\end{bmatrix}, \\
	\gamma  &\geq \mathbf{0}_{qr}.
	\end{aligned}
	\end{equation}
\end{lemma}

Note that any point in $\varPsi (\mathcal{Y}^i, F)$ is acceptable for achieving the resilient consensus. We could, thereby, choose any $\gamma^*$ satisfying \eqref{eqn:constraint}, and compute $y^i=\left(\mathbf{1}_{r}^{\prime} \otimes I_{d}\right) Y \gamma^*/r$. To this end, Phase I method in {\cite[Section~11.4]{boyd2004convex}} will be adopted by solving the following linear programming: 
\begin{equation}
\begin{aligned}
\max \qquad &\alpha\\
s.t. \qquad & \begin{bmatrix}\left(C \otimes I_{d}\right) Y \\
\left(I_{r} \otimes \mathbf{1}_{q}^{\prime}\right)
\end{bmatrix} \gamma =\begin{bmatrix}
\mathbf{0}_{dr} \\
\mathbf{1}_{r}
\end{bmatrix},\\
&\gamma_i\geq\alpha, \;\;i=1,2,...,qr.
\end{aligned}
\end{equation}
As $\varPsi (\mathcal{Y}^i, F)\neq\varnothing$, the problem is always solvable, with \revise{the computational complexity of \(\mathcal{O}\left((qr)^{3}\right)\)} (\hspace{1pt}\cite{boyd2004convex}). It is not surprising that this computational cost is independent of $|\mathcal{N}_i^+|$, as the cardinality of $\mathcal{Y}^i$ is fixed as $(d+1)F+1$. As a consequence, the algorithm will not introduce higher complexity in the network where some agent has a large number of in-neighbors. 

%%%%%%%%%%%%%%%%%%%%%%%%%%%%%%%%%%%%%%%%%%%%%%%%%%%%%%%%%%%%%%%%%%%%%%%
\subsection{Performance analysis}\label{sec:analysis}
In view of \eqref{eqn:updatewithattack}, the first two conditions in Theorem~\ref{thm:resilientCon} are naturally satisfied by noting that $\alpha=1/2$ and $\varepsilon^i(k)= 0, \forall i$. To prove the efficiency of Algorithm~\ref{alg:resilient}, the rest of this subsection aims to verify C3 and C4 of Theorem~\ref{thm:resilientCon}.

To this end, let us define the following set
\begin{equation}\label{def:S_i}
\mathcal S^i(k) \triangleq \varPsi (\mathcal{X}^i(k), F),
\end{equation}
where recall that $\mathcal{X}^i(k)$ is the collection of states received from in-neighbors at time $k$. We next introduce the following result:

\begin{lemma}\label{lmm:convexhull}
	Consider the network $\mathcal{G}=(\mathcal{V}, \mathcal{E})$. Suppose that the misbehaving agents follow either $F$-local or $F$-total attack model. With Algorithm \ref{alg:resilient}, it holds for any $i\in\mathcal{B}$ that
	\begin{enumerate}
	\item $y^i(p,k) \in \mathcal S^i(k)$ and $z^i(p,k) \in \mathcal S^i(k)$;
	\item $\mathcal{S}^i(k)$ is a subset of the convex hull formed by the states of benign agents, namely:
	\begin{equation}
		\mathcal{S}^i(k) \subset \Xi(k);
	\end{equation}
	\item \rm{C3} in Theorem~\ref{thm:resilientCon} is guaranteed. Moreover, 	
	\begin{equation}
		\Lambda^i(k) \subset \Xi(k),
	\end{equation} 
where $\Lambda^i(k)$ is defined in \eqref{eqn:Lambda^i}. 
\end{enumerate}
\end{lemma}

%\begin{remark}\label{rmk:safekernel}
%In Fig.~\ref{fig:safekernel}, we present a $2$D illustration of $\mathcal{S}^i(k)$. According to Definitions~\ref{def: S} and \ref{def: intersection}, it intuitively ignores the effects from the combination of any $F$ values. Therefore, it is ``safe'' in the sense that it must be contained in the convex hull formed by only benign states, as formally proved in the second statement of Lemma~\ref{lmm:convexhull}. We thus term $\mathcal{S}^i(k)$ as a ``safe kernel''. 
%\end{remark}

As a direct result of Lemma~\ref{lmm:convexhull}, the following corollary presents a non-expansion property of $\Xi(k)$: 
\begin{corollary}\label{thm:safety}
	Consider the network $\mathcal{G}=(\mathcal{V}, \mathcal{E})$. Suppose that the misbehaving agents follow either $F$-local or $F$-total attack model. With Algorithm \ref{alg:resilient}, the following relation holds at any $k\geq 0$:
	\begin{equation}
	\Xi(k+1)\subset \Xi(k).
	\end{equation}
\end{corollary}
\begin{pf*}{Proof.}
%At time step $k$, we focus on the `safe kernel', \textit{i.e.}, $\mathcal{S}^i(k)$, established by any healthy agent $i$. Lemma \ref{lm:existence} shows that $\mathcal{S}^i(k)$ will never be empty. 
Consider any benign agent $i$. By virtue of Lemma~\ref{lmm:convexhull}, one directly has $x^i(k+1)\in\Xi(k+1)$ as it is a convex combination of some points in $\Xi(k+1)$. Since this holds for any normal node, the proof is completed.
\hfill$\square$
\end{pf*}

Now we present sufficient conditions on network topology, under which the final condition in Theorem~\ref{thm:resilientCon} will be met:
\begin{lemma}\label{lmm:topo1}
	Consider the network $\mathcal{G}=(\mathcal{V}, \mathcal{E})$. Suppose that the network satisfies one of the following conditions:\\
	1) under $F$-local attack model, and is $((d+1)F+1)$-robust;\\
	2) under $F$-total attack model, and is $(dF+1,F+1)$-robust.\\
	With Algorithm \ref{alg:resilient}, {\rm C4} in Theorem~\ref{thm:resilientCon} is guaranteed. 
\end{lemma}

With the above results, one immediately concludes that the proposed algorithm facilitates the resilient consensus, as stated below:
\begin{theorem}\label{thm:converge}
	Consider the network $\mathcal{G}=(\mathcal{V}, \mathcal{E})$. Suppose that the network satisfies one of the following conditions:\\
	1) under $F$-local attack model, and is $((d+1)F+1)$-robust;\\
	2) under $F$-total attack model, and is $(dF+1,F+1)$-robust.\\
	With Algorithm \ref{alg:resilient}, it holds for any $k\in\mathbb{Z}_{\geq 0}$ and any $\ell\in\mathcal{D}$ that:
	\begin{equation}
	\Delta_\ell (k+|\mathcal{B}|)\leq \bigg(1-\frac{1}{2\cdot 4^{|\mathcal{B}|}}\bigg)\Delta_\ell(k),
	\end{equation}
	where $\Delta_\ell(k)$ is defined in \eqref{eqn:difference}. Moreover, benign agents exponentially achieve resilient consensus within the convex hull formed by their initial states, regardless of
	the actions of misbehaving ones. 
\end{theorem}
\begin{pf*}{Proof.}
	The theorem is immediately achieved by invoking Corollary~\ref{thm:safety} and substituting \eqref{eqn:delta} with $\alpha=0.5$ and $\delta(k)=0,\forall k$.
	\hfill$\square$
\end{pf*}

\begin{remark}
	By definitions, we note that any $((d+1)F+1)$-robust graph is $(dF+1,F+1)$-robust as well, but not vice versa. This is to say, a network which is able to tolerate $F$-local attacks could also survive the $F$-total attacks, while the converse may not be true. This observation is consistent with the fact that $F$-total attacks are special versions of $F$-local ones.
\end{remark}

By Theorem~\ref{thm:converge}, Algorithm~\ref{alg:resilient} guarantees that all benign agents reach an agreement on a weighted average of their initial states. Moreover, one concludes from Corollary~\ref{thm:safety} that the benign agents will never move out of $\Xi(0)$, i.e., the convex hull formed by their initial states. 

\begin{remark}\label{rmk:delta}\revise{
Let us recall the $\delta$-validity condition in Definition~\ref{def:resilientCon}. Notice that Algorithm~\ref{alg:resilient} enforces the validity condition with $\delta=0$, namely, the states of benign agents will converge to a point in $\Xi(0)$. However, as discussed before, this algorithm is computationally expense when $d$ or $F$ is large. Regarding this concern, if the system is willing to allow more performance loss, one can choose a larger $\delta$ and adopt some computationally efficient methods to approximate $y^i(p,k)$ and $z^i(p,k)$ instead of exactly calculating them. Hence, introducing $\delta$ in the validity condition also offers more freedom to the design of resilient consensus algorithms. }
\end{remark}

\subsection{\revise{Comparisons with existing works}}\label{sec:compare}
Before closing this section, we present discussions on comparing Algorithm~\ref{alg:resilient} with the existing resilient consensus algorithms. Specifically, our algorithm is shown to outperform the native approach, which independently applies the scalar algorithms to each component of vector states, in term of convergence accuracy. Moreover, as compared to other resilient vector consensus solutions, Algorithm~\ref{alg:resilient} has the advantage of reducing the computational cost.

\subsubsection{Comparisons with existing resilient consensus algorithms on scalar states}
In literature, there has been much work that proved to be effective in the simple scalar case (e.g., MSR in \cite{dolev1986reaching} and W-MSR in \cite{leblanc2013resilient}). Therefore, \revise{one might consider a naive way to achieve the resilient consensus in vector spaces by simply applying the existing scalar protocols to each component of the states through Kronecker product. Nevertheless, by doing so, the region that benign agents converge to is only guaranteed as a hypercube which contains their initial states. In contrast, Algorithm~\ref{alg:resilient} can achieve a sharper result with the decision vector exactly within the convex hull formed by these states. To see this, we present a $2$-dimensional illustration in Fig. \ref{Fig: boxvsconvexhull}. } 
\begin{figure}[!htbp]
	\centering
	\begin{tikzpicture}[scale=0.8]
\draw[color=gray!90!white,fill=green,fill opacity=0.2] (0,3.5)--(3,0)--(5,1.5)--(0,3.5);
%\draw[rounded corners,fill=gray,fill opacity=0.18] (-0.2,-0.2) rectangle (5.2,3.7);
\draw[fill=gray,dashed,fill opacity=0.18] (0,0) rectangle (5.,3.5);
\node[shape=circle,inner sep=0.5pt,fill=gray!50!black,fill opacity=1] (1) at (0,3.5) {$\;$};
\node[shape=circle,inner sep=0.5pt,fill=gray!50!black,fill opacity=1] (2) at (3,0) {$\;$};
\node[shape=circle,inner sep=0.5pt,fill=gray!50!black,fill opacity=1] (3) at (5,1.5) {$\;$};
\end{tikzpicture}%
	\caption{A $2$D illustration with agents marked with circles. The location of the node indicates its initial state. With the simple application of MSR or W-MSR to each dimension, the final agreement is ensured to be within the rectangle. In contrast, Algorithm~\ref{alg:resilient} offers a solution that lies in a sharper area, i.e., the green triangle.}
	\label{Fig: boxvsconvexhull}
\end{figure}

\revise{There are many examples where guaranteeing the solution within this convex hull is important. For instance, in some stochastic problem settings, each agent has a probability vector as its initial state. By forcing the decision vector to stay within the convex hull of their initial values, the benign agents can also agree on a probability vector.  As another example, for the convex optimization problems, the set
of feasible solutions is convex. Assuming that every benign agent
proposes a feasible solution, staying in this convex hull guarantees that the decision vector is also feasible. Using scalar consensus algorithm along each dimension, however, is insufficient to provide these guarantees in the examples above. Therefore, the problem of resilient vector consensus gains the research interests.}

\subsubsection{Comparisons with existing resilient consensus algorithms on vector states}
We next compare Algorithm~\ref{alg:resilient} with the existing solutions on resilient vector consensus. To this end, let us recall the updating rules in Algorithm~\ref{alg:resilient}. At every step, the normal agent $i$ obtains the states in its in-neighborhood, whereas up to $F$ of them might be faulty. To ensure its state is updated in a safe manner, agent $i$ hopes to use only good inputs, the cardinality of which is $|\mathcal{N}_i^+|-F$. Yet as it has no knowledge on the identities of these values, it seeks for \textit{resilient convex combination}, which is a point contained in the convex hull formed by \textit{any} $|\mathcal{N}_i^+|-F$ neighboring states.

Some works achieve this through calculating Tverberg points (\cite{vaidya2013byzantine,vaidya2014iterative}). While the results therein are elegant, the calculation of Tverberg points is rather costly and almost impossible in many cases (\hspace{1pt}\cite{agarwal2008algorithms}). The existing works unfortunately do not provide an efficient way to do it, which leads to a major concern in applying Tverberg points to facilitate resilient consensus. Hence, instead of using Tverberg points, this paper, inspired by Helly's Theorem, \revise{achieves the resilient convex combination by leveraging the intersection of a set of convex hulls, namely, $\mathcal{S}^i(k)$ as defined in \eqref{def:S_i}.} 
	
\revise{In Fig.~\ref{fig:safekernel}, we present a $2$D illustration of $\mathcal{S}^i(k)$. According to Definitions~\ref{def: S} and \ref{def: intersection}, it intuitively ignores the effects from the combination of any $F$ values. Therefore, it is ``safe'' in the sense that it must be contained in the convex hull formed by only benign states, as formally proved in the second statement of Lemma~\ref{lmm:convexhull}. We thus term $\mathcal{S}^i(k)$ as a ``safe kernel''. At any time, the healthy agent modifies its state towards this kernel. The effects of malicious agents are thus limited.}
	
	\begin{figure}[!htbp]
		\centering
		\begin{tikzpicture}[fill opacity=0.1]

\begin{axis}[%
width=1.8in,
height=1.2in,
color= white,
scale only axis,
xtick=\empty,
xmin=-0.2,
xmax=2,
ytick=\empty,
ymin=-0.2,
ymax=4.1,
axis lines=middle, 
]
\draw[color=black,fill=gray,fill opacity=0.03] (0,3)--(0.4,0.7)--(1.5,0)--(1,4)--(0,3);
\draw[color=black,fill=gray,fill opacity=0.03] (0,3)--(0.4,0.7)--(1.5,0)--(1.8,2.2)--(0,3);
\draw[color=black,fill=gray,fill opacity=0.03] (0,3)--(0.4,0.7)--(1.8,2.2)--(1,4)--(0,3);
\draw[color=black,fill=gray,fill opacity=0.03] (0,3)--(1.5,0)--(1.8,2.2)--(1,4)--(0,3);
\draw[color=black,fill=gray,fill opacity=0.03] (0.4,0.7)--(1.5,0)--(1.8,2.2)--(1,4)--(0.4,0.7);
\node[shape=circle,inner sep=1.0pt,fill=gray!50!black,fill opacity=1] (1) at (0,3) {$\;$};
\node[shape=circle,inner sep=1.0pt,fill=gray!50!black,fill opacity=1] (2) at (0.4,0.7) {$\;$};
\node[shape=circle,inner sep=1.0pt,fill=gray!50!black,fill opacity=1] (3) at (1.5,0) {$\;$};
\node[shape=circle,inner sep=1.0pt,fill=gray!50!black,fill opacity=1] (4) at (1.8,2.2) {$\;$};
\node[shape=circle,inner sep=1.0pt,fill=gray!50!black,fill opacity=1] (5) at (1,4) {$\;$};
\draw[color=black,fill=green,fill opacity=0.2] (0.6,1.8)--(0.7570,2.6635)--(1.1912,2.4706)--(1.2929,1.6568)--(0.8883,1.2233)--(0.600,1.8);
%\draw[dashed, color=cyan,line width=1.1] (0.4,0.7)--(0.4,3)--(1.5,3)--(1.5,0.7)--(0.4,0.7);
%\node[shape=circle,inner sep=1.0pt,fill=blue,fill opacity=1] at (0.6,1.8) {$\;$};
%\node[shape=diamond,inner sep=1.0pt,fill=purple,fill opacity=1] at (1.2929,1.6568) {$\;$};

%\node[left, align=left]
\end{axis}
\end{tikzpicture}%
		\caption{A $2$D illustration of ``safe kernel''. Suppose that agent $i\in\mathcal{B}$ has $5$ in-neighbors and each of their states is represented by a black circle. Let $F=1$. The green region denotes the safe kernel $\mathcal{S}^i(k)=\varPsi(\mathcal{X}^i(k), F)$.}
		\label{fig:safekernel}
	\end{figure}
	
\revise{Notice that, in \cite{mendes2013multidimensional} and \cite{yan2019safe}, the authors require the safe kernel to be exactly calculated. However, since this kernel is the intersection of $\binom{|\mathcal{N}_i^+|}{F}$ convex hulls, the existing approaches for computing the kernel are usually difficult (\cite{rademacher2007approximating}). Moreover, as the computation of $\mathcal{S}^i(k)$ depends on every state in $\mathcal{X}^i(k)$, their algorithms may fail when node $i$ has a large number of in-neighbors. In contrast, instead of computing the whole kernel, only $2d$ points in it are calculated in Algorithm~\ref{alg:resilient}. As shown in Section~\ref{sec:middle_points}, this step can be achieved by linear programming, the computational complexity of which will not increase in a larger neighborhood. Hence, our algorithm yields a lower cost than these existing works.}

\section{A modification of Algorithm~\ref{alg:resilient}}\label{sec:algo2}
Despite achieving a decent performance in the presence of attacks/faults, note that Algorithm~\ref{alg:resilient} can only guarantee a sub-optimal performance in benign environment\footnote{\revise{To be more specific, this refers to the case of ``false alarm''. Namely, the benign agents believe there exist $F>0$ faulty nodes and thus perform Algorithm~\ref{alg:resilient} with parameter $F$. However, in fact no faulty ones exist and the network is operating in the benign environment. In this case, each agent, by performing Algorithm~\ref{alg:resilient}, will actively abandon some information it received (even though they are benign) as it regards them as misleading.}}, namely when $|\mathcal{F}|=0$. To be particular, it only ensures converging to an agreement within $\Xi(0)$ instead of to the exact average of the agents' initial states. This means that in order to increase the system's resiliency, Algorithm~\ref{alg:resilient} sacrifices the system's performance during normal operations.

To cope with this issue, \revise{this section seeks an alternative solution that equips the standard average consensus algorithms with an additional security guarantee.} In particular, given any standard average consensus algorithm, a modification of Algorithm~\ref{alg:resilient} is proposed, which switches between the standard one in the absence of faulty nodes and Algorithm~\ref{alg:resilient} in the presence of false nodes.

For ease of illustration, let us take the linear update rule as an example\footnote{The results can be readily generalized to enhance the security of other standard average consensus algorithms such as the ones using nonlinear updates or tackling switching topology.}. Suppose that the following protocol solves the problem of average consensus:
\begin{equation}\label{eqn:linear}
x^i(k+1) = a^{ii} x^i(k) + \sum_{j\in\mathcal{N}_i^+} a^{ij} x^j(k)
\end{equation}
such that $a^{ii}+\sum_{j\in\mathcal{N}_i^+} a^{ij}=1$ and each weight lowered bounded by some $\alpha\in(0,1)$. Given \eqref{eqn:linear}, we shall show how to enhance the security of it.

To this end, let us define $$\tilde{a}^{ij} \triangleq \frac{a^{ij}}{1-a^{ii}}.$$
Based on \eqref{eqn:linear}, a resilient consensus strategy is presented in Algorithm~\ref{alg:switch}.

\begin{algorithm}
	1:\: Receive the states from all in-neighbors $j\in\mathcal{N}^+_i$, and collect these values in $\mathcal{X}^i(k)$.\\
	2:\: Agent $i$ calculates the weighted average of the received states as
	\begin{equation}\label{eqn:alg2}
	\bar{x}^i(k) = \sum_{j\in\mathcal{N}_i^+}\tilde{a}^{ij}x^j(k).
	\end{equation}
	3:\: Depending on the distance among received states, let us define 
	\begin{align*}\label{eqn:switch}	
	\hspace{-100pt}
		\breve{x}^i(k)\triangleq
\end{align*}
	\begin{numcases}{}
		\bar{x}^i(k),  \text{ if } ||x^j(k)-\bar{x}^i(k)||\leq \lambda(k), \forall j\in\mathcal{N}_i^+, \label{eqn:ave}\\
	\cen(\Lambda^i(k)),  \text{ otherwise}, \label{eqn:resilient}
	\end{numcases}
	where $\Lambda^i(k)$ is calculated from Steps 2--3 in Algorithm~\ref{alg:resilient} and                                                                                                                                                                                                                                                                                                                                                                                                                                                                                                                                                                                                                                                                                                                                                                                                                                                                                                                                                                                                                                                                                                                                                                                                                                                                                                                                                                                                                                                                                                                                                                                                                                                                                                                                                                                                                                                                                                                                                                                                                                                                                                                                                                                                                                                                                                                                                                                                                                                                                                                                                                                                                                                                                                                                                                                                                                                                                                    the threshold is set as $\lambda(k)=c\sigma^{k}$ with $c>0$ and  $0<\sigma<1$. Agent $i$ updates its local state by the following law:
	\begin{equation}\label{eqn:switch1}
	x^i(k+1) = a^{ii}x^i(k)+(1-a^{ii})\breve{x}^i(k).
	\end{equation}
	4:\: Transmit the new state $x^i(k+1)$ to out-neighbors $j\in\mathcal{N}_i^-$. 
	\caption{Resilient consensus algorithm $2$ for benign agent $i$}
	\label{alg:switch}
\end{algorithm}

Notice that the key idea of Algorithm~\ref{alg:switch} is to ``resiliently" compute the local average at each step.
Intuitively, when the difference among received states is small, we regard the effects of misbehaving ones (if any) as negligible and adopt the local average $\bar{x}^i(k)$ for updates. Otherwise, as proved in Section~\ref{sec:alg}, $\cen(\Lambda^i(k))$ offers a safe way to replace $\bar{x}^i(k)$ and prevents benign states from being significantly affected. This switching rule plays a vital role in enhancing the resiliency in adversarial environment while guaranteeing the optimality when no faulty agents exist, as will be proved soon.

\subsection{Performance analysis}
In this subsection, we shall respectively analyze the performance of Algorithm~\ref{alg:switch} in the presence and absence of faulty nodes.

\subsubsection{Performance in the presence of faulty nodes}
Let us consider the scenario where the faulty agents exist. To begin with, we shall first provide an equivalent form of the updating rule in Algorithm~\ref{alg:switch}.
\begin{lemma}\label{lmm:equvi}
	Consider the network $\mathcal{G}=(\mathcal{V}, \mathcal{E})$. Suppose that the benign agent $i\in\mathcal{B}$ updates with Algorithm \ref{alg:switch}. Then the update of it \eqref{eqn:alg2}--\eqref{eqn:switch1} is equivalent to the following form at any $k\geq 0$:
	\begin{equation}\label{eqn:eqv}
	\begin{split}
	x^i(k+1) = a^{ii} x^i(k)+ (1-a^{ii} )\cen(\Lambda^i(k))  + \varepsilon^i(k),
	\end{split}
	\end{equation}
	where $\lim_{k\rightarrow \infty}\varepsilon^i(k)=0$ and $\sum_{k=0}^\infty ||\varepsilon^i(k) ||<\infty$. 
\end{lemma}

For convenience, let us set $\tilde{x}^i(k) \triangleq \cen(\Lambda^i(k)).$ In view of \eqref{eqn:eqv}, the new state is a convex combination of its current state $x^i(k)$ and an auxiliary point $\tilde{x}^i(k)$, summed by an exponential decaying term $\varepsilon^i(k)$. Leveraging \eqref{eqn:eqv}, it is not difficult to check that all conditions in Theorem~\ref{thm:resilientCon} are met. Therefore, the resiliency of Algorithm~\ref{alg:switch} is verified. We formally state the conclusion as below:
\begin{theorem}\label{thm:converge2}
	Consider the network $\mathcal{G}=(\mathcal{V}, \mathcal{E})$. Suppose that the network satisfies one of the following conditions:\\
	1) under $F$-local attack model, and is $((d+1)F+1)$-robust;\\
	2) under $F$-total attack model, and is $(dF+1,F+1)$-robust.\\
	With Algorithm \ref{alg:switch}, benign agents exponentially achieve resilient consensus within the following set, regardless of
	the actions of misbehaving ones: 
	\begin{equation}\label{eqn:bound}
	\Xi \triangleq \bigg\{x:||x-\chi|| \leq \frac{c(1-\alpha)}{1-\sigma}, \exists \chi\in \Xi(0)\bigg\}.
	\end{equation}
\end{theorem}
\begin{pf*}{Proof.}
	Notice that 
	\begin{equation}
	\begin{split}
	\delta(k) &\triangleq \max_{i\in\mathcal{B}}||\varepsilon^i(k)||\leq (1-\alpha)\lambda(k)=(1-\alpha)c\sigma^k.
	\end{split}
	\end{equation}
	Recalling Lemma~\ref{lmm:H_i}, the theorem can be obtained. 
	\hfill$\square$
\end{pf*}

The threshold in Algorithm \ref{alg:switch}, i.e. $\lambda(k)$, can be interpreted as the disagreement of states that a benign agent is willing to tolerate in its neighborhood. Increasing $\lambda(k)$ leads to a higher possibility to use the local average $\bar{x}^i(k)$ in updates, and thus results in the range where the final agreement converges to, i.e. $\Xi$, being less sharper during attacks. However, as to be proved later, by choosing a ``large enough'' $\lambda(k)$, Algorithm~\ref{alg:switch}  is capable of accounting for the state disagreement in benign environment and thus achieving the optimal performance during normal operation.

\begin{remark}
Notice that \eqref{eqn:bound} presents the worst-case upper bound, which may be too conservative and can never be touched in the practical scenario. Therefore, how to infer a tighter bound will be left as our future work.
\end{remark}

\begin{remark}
	Compared with the simple approach of directly using the existing scalar algorithms (like MSR or W-MSR) to each dimension of the state vectors, \revise{Algorithms~\ref{alg:resilient} and \ref{alg:switch} have an obvious limitation since they inevitably introduce higher computational costs, especially when $d$ or $F$ is large.} Notice that this is a fundamental limitation of any resilient vector consensus algorithms in this line of research, including \cite{vaidya2013byzantine,vaidya2014iterative,yan2019safe,mendes2013multidimensional}. However, on the other hand, our work would provide a better convergence result as shown in Fig. \ref{Fig: boxvsconvexhull}. This fact indicates a trade-off exists between high consensus accuracy and low computational complexity. Our work provides a fine tuning of this trade-off. For example, in practice, it might be the case that certain components of system variables represent more critical information. In this case, one can partition $d$ dimensions into several subgroups and apply Algorithm~\ref{alg:resilient} or \ref{alg:switch} on the group of these critical ones to better protect the system, while applying MSR or W-MSR to the other components to reduce the computational burden. 
\end{remark}

\subsubsection{Performance in the absence of faulty nodes}
This part investigates the performance of Algorithm~\ref{alg:switch} in benign environment. We will show that this performance agrees with that of the standard rule \eqref{eqn:linear}. Therefore, even under the case of false alarm, where the agents mistakenly assume that there exist $F$ faulty ones (but indeed there are not), the exact average consensus can still be achieved by performing Algorithm~\ref{alg:switch}.

To see this, let us define the following matrix:
\begin{equation}
A\triangleq [a^{ij}], \; \revise{L \triangleq I-A, \;\hat L \triangleq \frac{L+L^T}{2}}.
\end{equation}
We denote by $\mu_2$ the Fiedler eigenvalue of $\hat L$, i.e., $$\mu_2 \triangleq \min _{v \neq 0 \atop \1^{\mathrm{T}} v=0} \frac{v^{T} \hat L v}{\|v\|^{2}}.$$ The following theorem provides conditions under which Algorithm~\ref{alg:switch} coincides with the conventional average consensus protocol \eqref{eqn:linear} and achieves the optimal solution in fault-free environment:

\begin{theorem}\label{thm:benign}
	Consider the network $\mathcal{G}=(\mathcal{V}, \mathcal{E})$. Suppose that each agent updates with Algorithm~\ref{alg:switch}, where
	$c \geq (1-\alpha)\max_{i,j\in\mathcal{V}} ||x^i(0)-x^j(0)||$ and $\sigma \geq \mu_2.$ When no misbehaving agent exists, Algorithm \ref{alg:switch} is equivalent to \eqref{eqn:linear}. In particular, it facilitates the average consensus among agents in benign environment. That is, it holds that 
	\begin{equation}
	\lim_{k\rightarrow \infty} x^i(k) =\frac{1}{N}\sum_{i=1}^{N}x^i(0), \;\forall i\in\mathcal{V}.
	\end{equation}
\end{theorem}
\begin{pf*}{Proof.}
	Consider the update of agent $i$ at any time $k$. Invoking \cite{olfati2004consensus}, the difference between any in-neighboring state $x^j(k)$ ($\forall j\in\mathcal{N}_i^+$) and the local weighted average $\bar{x}^i(k)$ will not exceed $\lambda(k)$ during normal operation. Therefore, any agent always uses $\bar{x}^i(k)$ for updates and Algorithm~\ref{alg:switch} is equivalent to \eqref{eqn:linear}. As assumed, \eqref{eqn:linear} solves the problem of average consensus. The proof thus follows.
	\hfill$\square$
\end{pf*}

	\begin{remark}
		\revise{
		The upper bound of $||x_i(0)-x_j(0)||$ can be estimated in many practical systems, since the states of agents represent quantities that are physically constrained (\cite{vaidya2013byzantine}). For instance, if the states are $2$-dimensional stochastic vectors, it is calculated that $||x_i(0)-x_j(0)||\leq \sqrt{2}, \;\forall i,j\in\mathcal{V}$. If the states represent locations of mobile robots in a $3$-dimensional space, then $\max_{i,j\in\mathcal{V}} ||x_i(0)-x_j(0)||$ is determined by boundary of the region where robots are allowed to operate.}
	\end{remark} 

Therefore, Algorithm~\ref{alg:switch} maintains the optimal performance of the standard consensus algorithm during normal operation while raising its level of security in the presence of faults/attacks. Compared with Algorithm~\ref{alg:resilient}, where only a weighted average consensus is reached by agents, it achieves the exact average in benign environment, at the expense of sacrificing the performance in the presence of attacks, which may lead the benign agents to move out of the convex hull formed by their initial states. We should notice that the trade-off between the performance with and without attacks can be balanced through the threshold $\lambda(k)$.

\section{Numerical Example}\label{sec:casestudy}
In this section, we provide some numerical examples to verify the theoretical results established in the previous sections. 

\subsection{$2$-dimensional example}
We choose the $(3,2)$-robust \revise{digraph}\footnote{\revise{This network is established based on the robust digraph given in \cite{dibaji2017resilient}, by following the method proposed in \cite{leblanc2012consensus} which shows how to construct an $(r,s)$-robust digraph from an existing one.}} in Fig. \ref{Fig:network} as communication network, where the node set is $\mathcal{V} = \{1, \cdots,6\}$. According to Theorems~\ref{thm:converge} and \ref{thm:converge2}, the network can tolerate a single misbehaving node in a $2$-dimensional system. Suppose that this misbehaving agent intends to prevent others from reaching a correct consensus by setting its states as $x^2_1(k) = 4.5*\sin(k/5)$ and $x^2_2(k) = k/25+1$. On the other hand, the benign agents are initialized randomly within $[-1,1.5]\times[0.5,4]$.

 \begin{figure}[!htbp]
 \centering
 \begin{tikzpicture}[scale =0.38] 
\coordinate (O) at (0, 0);
\node[shape=circle,inner sep=4pt,fill=lightgray,fill opacity=0.7] (1) at (-5,4) {{\small $1$}};
\node[shape=circle,inner sep=4pt,fill=red,fill opacity=0.8] (2) at (0,4) {{\small $2$}};
\node[shape=circle,inner sep=4pt,fill=lightgray,fill opacity=0.7] (3) at (5,4) {{\small $3$}};
\node[shape=circle,inner sep=4pt,fill=lightgray,fill opacity=0.7] (4) at (-5,0) {{\small $4$}};
\node[shape=circle,inner sep=4pt,fill=lightgray,fill opacity=0.7] (5) at (0,0) {{\small $5$}};
\node[shape=circle,inner sep=4pt,fill=lightgray,fill opacity=0.7] (6) at (5,0) {{\small $6$}};

\draw[->] (1)--(-6.2,4)--(-6.2,5.5)--(5,5.5)--(3);
\draw[->] (6)--(6.2,0)--(6.2,-1.5)--(-5,-1.5)--(4);
\draw[->] (2) edge (3) (5) edge (4);
\draw[<->] (1)--(6);
\draw[<->] (3)--(4);
\draw[<->] (1)--(4);
\draw[<->] (2)--(5);
\draw[<->] (1) edge (5) (5) edge (3);
\draw[<->] (4)--(2);
\draw[<->] (2)--(6);
\draw[<->] (1)--(2);
\draw[<->] (3)--(6);
\draw[<->] (6)--(5);
\end{tikzpicture}
 \caption{\modify{A $(3,2)$-robust communication network.}}
 \label{Fig:network}
 \end{figure}

% \begin{figure}[!htbp]
%	\centering
%	\input{tikz/agreement.tikz}
%	\caption{\revise{The states of normal agents $x^i(k), \forall i\in\mathcal{B}$.}}
%	\label{fig:agreement}
%\end{figure}
The performance of Algorithm~\ref{alg:resilient} is tested in Figs.~\ref{fig:trajectory} and \revise{\ref{fig:agreement}}. Results show that the benign agents are guaranteed to stay in the convex hull of their initial states at any time and they finally reach an agreement, which validates Theorem~\ref{thm:converge}. As compared to the native approach of independently applying W-MSR (\cite{leblanc2013resilient}) to each dimension of the vector state, our strategy yields a more accurate consensus result.

\begin{figure}
	\centering
	\begin{subfigure}[b]{0.4\textwidth}
		\input{tikz/helly.tikz}
		\caption{Trajectory of local states under Algorithm~\ref{alg:resilient}.}
	\end{subfigure}
\centering
	\begin{subfigure}[b]{0.4\textwidth}
		\input{tikz/MSR.tikz}
		\caption{Trajectory of local states under using W-MSR (\cite{leblanc2013resilient}) independently to each entry.}
	\end{subfigure}
	\caption{\modify{Comparison of local states' trajectory under different algorithms, where the area surrounded by the dashed lines is the convex hull formed by the initial states of benign agents.}}
	\label{fig:trajectory}
\end{figure}

\begin{figure}[!htbp]
	\centering
\definecolor{mycolor1}{rgb}{0.00000,0.44700,0.74100}%
\definecolor{mycolor2}{rgb}{0.85000,0.32500,0.09800}%
\definecolor{mycolor3}{rgb}{0.92900,0.69400,0.12500}%
\definecolor{mycolor4}{rgb}{0.49400,0.18400,0.55600}%
\definecolor{mycolor5}{rgb}{0.46600,0.67400,0.18800}%
\definecolor{mycolor6}{rgb}{0.30100,0.74500,0.93300}%
\definecolor{mycolor7}{rgb}{0.63500,0.07800,0.18400}%
\begin{tikzpicture}[scale=0.8]

\begin{axis}[%
width=2.2in,
height=1.5in,
at={(1.145in,3.292in)},
scale only axis,
xmin=1,
xmax=15,
ymin=-2,
ymax=2.6,
ylabel={$x^i_1(k)$},
axis background/.style={fill=white}
]
\addplot [color=mycolor1, line width=1.0pt ]
  table[row sep=crcr]{%
1	-0.294\\
2	0.499994875\\
3	0.50733175\\
4	0.418881106\\
5	0.440868388\\
6	0.441135439\\
7	0.440882064\\
8	0.44059675\\
9	0.440656634\\
10	0.44066776\\
11	0.440658383\\
12	0.440660726\\
13	0.440660574\\
14	0.440660471\\
15	0.440660456\\
};
\addplot [color=mycolor2, line width=1.0pt, dashed]
  table[row sep=crcr]{%
1   0.49667\\
2	0.973545856	\\
3	1.411606183	\\
4	1.793390227	\\
5	2.103677462	\\
6	2.330097715	\\
7	2.463624325	\\
8	2.498934008	\\
9	2.434619077	\\
10	2.273243567	\\
11	2.02124101	\\
12	1.688657951	\\
13	1.28875343	\\
14	0.837470375	\\
15	0.35280002	\\
16	-0.145935359	\\
17	-0.638852755	\\
18	-1.106301108	\\
19	-1.529644727	\\
20	-1.892006238	\\
21	-2.178939431	\\
22	-2.379005185	\\
23	-2.484227509	\\
24	-2.490411522	\\
25	-2.397310687	\\
26	-2.208636639	\\
27	-1.931911219	\\
28	-1.578166595	\\
29	-1.161505449	\\
30	-0.698538745	\\
31	-0.207723507	\\
32	0.291373012	\\
33	0.778853409	\\
34	1.235283378	\\
35	1.642466497	\\
36	1.98416966	\\
};
\addplot [color=mycolor3, line width=1.0pt ]
  table[row sep=crcr]{%
1	-1.099\\
2	-0.16633329\\
3	0.285222361\\
4	0.430344668\\
5	0.440951658\\
6	0.440896135\\
7	0.440257411\\
8	0.440673868\\
9	0.440623518\\
10	0.440645608\\
11	0.440660374\\
12	0.440660678\\
13	0.440660559\\
14	0.44066045\\
15	0.440660501\\
};
\addplot [color=mycolor4, line width=1.0pt ]
  table[row sep=crcr]{%
1	0\\
2	0.3\\
3	0.278005416\\
4	0.428158695\\
5	0.410473802\\
6	0.4308935\\
7	0.437555549\\
8	0.439851063\\
9	0.440384252\\
10	0.440578325\\
11	0.440631221\\
12	0.440650767\\
13	0.44065727\\
14	0.440659545\\
15	0.440660232\\
};
\addplot [color=mycolor6, line width=1.0pt ]
  table[row sep=crcr]{%
1	1.468\\
2	0.491333449\\
3	0.455670747\\
4	0.455282268\\
5	0.442775681\\
6	0.440344382\\
7	0.440609585\\
8	0.440571101\\
9	0.440614937\\
10	0.440640136\\
11	0.440651826\\
12	0.44065718\\
13	0.440659519\\
14	0.440660222\\
15	0.440660369\\
};
\addplot [color=mycolor7, line width=1.0pt]
  table[row sep=crcr]{%
1	2.121\\
2	0.511000188\\
3	0.355670341\\
4	0.450477102\\
5	0.435106751\\
6	0.438947849\\
7	0.440406225\\
8	0.440723455\\
9	0.440690372\\
10	0.440657391\\
11	0.440663227\\
12	0.440660491\\
13	0.44066063\\
14	0.440660441\\
15	0.440660499\\
};
\end{axis}

\begin{axis}[%
width=2.2in,
height=1.5in,
at={(1.145in,1.65in)},
scale only axis,
xmin=1,
xmax=15,
xlabel={Time $k$},
ymin=0.5,
ymax=3.8,
ylabel={$x^i_2(k)$},
axis background/.style={fill=white}
]
\addplot [color=mycolor1, line width=1.0pt]
table[row sep=crcr]{%
1	2.574\\
2	2.1058883968468\\
3	2.01351815642558\\
4	2.04555214384146\\
5	2.04372837824908\\
6	2.04175140235553\\
7	2.04259083419923\\
8	2.04268158518451\\
9	2.04266463614335\\
10	2.04265966628948\\
11	2.04266324551695\\
12	2.04266146884567\\
13	2.04266198968829\\
14	2.04266215416487\\
15	2.04266207132226\\
};
\addplot [color=mycolor2, line width=1.0pt, dashed]
table[row sep=crcr]{%
1	1.04\\
2	1.08\\
3	1.12\\
4	1.16\\
5	1.2\\
6	1.24\\
7	1.28\\
8	1.32\\
9	1.36\\
10	1.4\\
11	1.44\\
12	1.48\\
13	1.52\\
14	1.56\\
15	1.6\\
16	1.64\\
17	1.68\\
18	1.72\\
19	1.76\\
20	1.8\\
21	1.84\\
22	1.88\\
23	1.92\\
24	1.96\\
25	2\\
26	2.04\\
27	2.08\\
28	2.12\\
29	2.16\\
30	2.2\\
31	2.24\\
32	2.28\\
33	2.32\\
34	2.36\\
35	2.4\\
36	2.44\\
};
\addplot [color=mycolor3, line width=1.0pt ]
table[row sep=crcr]{%
1	3.695\\
2	2.49366657332073\\
3	2.1427775479101\\
4	2.05697470196246\\
5	2.03760197956847\\
6	2.04168624743094\\
7	2.04280054966699\\
8	2.04266075106358\\
9	2.04267423767544\\
10	2.04266783735285\\
11	2.04266239182309\\
12	2.04266147734486\\
13	2.04266198977978\\
14	2.04266203339511\\
15	2.04266201956753\\
};
\addplot [color=mycolor4, line width=1.0pt ]
table[row sep=crcr]{%
1	1.890\\
2	1.893\\
3	2.06512782494439\\
4	2.08761996897659\\
5	2.06514540078165\\
6	2.05497886547602\\
7	2.0465514542971\\
8	2.04343591770597\\
9	2.04291623578274\\
10	2.04274342306381\\
11	2.04269008762986\\
12	2.0426714512317\\
13	2.04266508711906\\
14	2.04266272797608\\
15	2.04266220745817\\
};
\addplot [color=mycolor6, line width=1.0pt ]
table[row sep=crcr]{%
1	3.058\\
2	2.5083331371173\\
3	2.18557143395601\\
4	2.09046466674599\\
5	2.05585844992139\\
6	2.0468754263463\\
7	2.04370274058469\\
8	2.04302726586222\\
9	2.04278736754517\\
10	2.04270569522235\\
11	2.04267734437766\\
12	2.04266724042446\\
13	2.04266342660069\\
14	2.04266244017708\\
15	2.04266216619203\\
};
\addplot [color=mycolor7, line width=1.0pt ]
table[row sep=crcr]{%
1	0.754\\
2	1.96733303424721\\
3	2.03252956143995\\
4	2.02352470360752\\
5	2.04387941364329\\
6	2.04377870039205\\
7	2.0424271657868\\
8	2.04253630511216\\
9	2.04261929086381\\
10	2.04265394065565\\
11	2.04265862370831\\
12	2.04266226059618\\
13	2.04266173709077\\
14	2.04266201786029\\
15	2.04266202083059\\
};
\end{axis}
\end{tikzpicture}%
	\caption{\modify{The states of agents under Algorithm~1, where the red dashed line represent the state of faulty agent, and the others represent those of benign ones.}}
	\label{fig:agreement}
\end{figure}

We next compare the performance of Algorithm~\ref{alg:switch} with those of 1) Algorithm~\ref{alg:resilient}, 2) the standard average consensus algorithm given in \cite{olfati2004consensus}, and 3) the native approach of simply applying W-MSR to each entry. In \cite{olfati2004consensus}, the adjacent matrix is set as
\begin{equation}
\modify{A = \begin{bmatrix}
0.3  &0.2    &0   &0.2 &0.15  &0.15\\
0.15 &0.25   &0   &0.25 &0.15  &0.2\\
0.1  &0.1   &0.2  &0.25  &0.25  &0.1\\
0.15 &0.1   &0.2  &0.3  &0.15  &0.1\\
0.15 &0.15  &0.4   &0   &0.15  &0.15\\
0.15 &0.2   &0.2   &0   &0.15  &0.3
\end{bmatrix},}
\end{equation}
where each weight is lower bounded by $\alpha=0.1$. To enforce the conditions in Theorem~\ref{thm:benign}, the parameters in Algorithm~\ref{alg:switch} are set as $c=4.5$ and \revise{$\sigma = 0.6$}.

As shown in Figs.~\ref{fig:woattack} and \ref{fig:wattack}, in the fault-free environment, Algorithm~\ref{alg:switch} achieves the same performance as that of \cite{olfati2004consensus}, while Algorithm~\ref{alg:resilient} and W-MSR, although can guarantee an agreement among agents, are unable to converge to the exact average of agents' initial states. On the other hand, Algorithm~\ref{alg:switch} performs similarly to other resilient solutions to provide additional security guarantees in the presence of network misbehaviors. 

%Moreover, Fig shows the $\delta$-Validity condition is also satisfied by our solution.

\begin{figure}
	%\centering
	\begin{subfigure}[b]{0.45\textwidth}
\definecolor{mycolor1}{rgb}{0.00000,0.44700,0.74100}%
\definecolor{mycolor2}{rgb}{0.85000,0.32500,0.09800}%
\definecolor{mycolor3}{rgb}{0.92900,0.69400,0.12500}%
\definecolor{mycolor4}{rgb}{0.63500,0.07800,0.18400}%
\begin{tikzpicture}[scale=0.9]

\begin{semilogyaxis}[%
width=1.6in,
height=1.4in,
at={(1.011in,0.642in)},
scale only axis,
xmin=0,
xmax=20,
xlabel = {$k$},
ymin=-20,
ymax=30,
ylabel = {$\sum_{i\in\mathcal{V}}||x^i(k)-\bar{x}(k)||$},
legend entries={{\small Algorithm~\ref{alg:resilient}}, {\small Algorithm~\ref{alg:switch}}, {\small Standard algorithm}, {\small W-MSR}},
legend style={legend cell align=left, align=left, draw=white!15!black},
legend pos={outer north east}
]
\addplot[color=mycolor4, line width=1.0pt, dotted]
table[row sep=crcr]{%
1	6.35718381667918\\
2	2.85948897011215\\
3	1.32521266565948\\
4	0.649662589994511\\
5	0.321525450411105\\
6	0.159761463762685\\
7	0.07983695155899\\
8	0.0400563719748802\\
9	0.0200807968069406\\
10	0.0100586182600901\\
11	0.00503653526843353\\
12	0.00252164547807273\\
13	0.00126284995434225\\
14	0.000632758991836008\\
15	0.000318611084265647\\
16	0.000159953266270165\\
17	8.00644924459933e-05\\
18	4.01393650979275e-05\\
19	2.01341511113403e-05\\
20	1.00492102943383e-05\\
21	5.01958596370735e-06\\
22	2.4863727726501e-06\\
23	1.23265380335821e-06\\
24	6.03720276362196e-07\\
25	2.94696844925555e-07\\
26	1.41962111880397e-07\\
27	6.69594993288413e-08\\
28	3.07915927199182e-08\\
29	1.42534472124256e-08\\
30	6.41226145716984e-09\\
31	2.93283687847668e-09\\
32	1.33604317128865e-09\\
33	6.12935728224248e-10\\
34	2.83309903993281e-10\\
35	1.30141372971036e-10\\
36	5.87438990148783e-11\\
37	2.6624635940828e-11\\
38	1.20125053816255e-11\\
39	5.37504471558182e-12\\
40	2.41758029116506e-12\\
41	1.08113994277826e-12\\
42	4.86012455851995e-13\\
43	2.16569342483308e-13\\
44	9.80154189216547e-14\\
45	4.37473342971193e-14\\
46	2.01170746923165e-14\\
47	8.9194914440109e-15\\
48	4.04670992233117e-15\\
49	2.1591240957302e-15\\
50	1.24745413748737e-15\\
};

\addplot [color=mycolor1, line width=1.0pt]
table[row sep=crcr]{%
1	8.99104678423662\\
2	2.25694732588403\\
3	0.763234903570932\\
4	0.268185641756166\\
5	0.102994273803756\\
6	0.0365124958063614\\
7	0.0141662585684589\\
8	0.00508871186911525\\
9	0.0019782558031438\\
10	0.000720193366498625\\
11	0.000280005343723297\\
12	0.000103156040760475\\
13	4.00760003266941e-05\\
14	1.49136007026009e-05\\
15	5.78709833201552e-06\\
16	2.17142134476118e-06\\
17	8.41483882862065e-07\\
18	3.17831859560565e-07\\
19	1.23007923757668e-07\\
20	4.67017727774631e-08\\
21	1.80533248440495e-08\\
22	6.88162986054325e-09\\
23	2.65753186432298e-09\\
24	1.01608655935394e-09\\
25	3.92069540339965e-10\\
26	1.50247364599441e-10\\
27	5.79373697757046e-11\\
28	2.22391551542192e-11\\
29	8.57185143413095e-12\\
30	3.29471956567528e-12\\
31	1.26897896366241e-12\\
32	4.88495723336086e-13\\
33	1.88306315049239e-13\\
34	7.22397249882379e-14\\
35	2.78209003179767e-14\\
36	1.1151288555545e-14\\
37	4.0808192092798e-15\\
38	2.09188894695423e-15\\
39	3.88578058618805e-16\\
40	1.99497218720156e-15\\
41	1.98979322135008e-15\\
42	1.33313414590556e-15\\
43	1.33313414590556e-15\\
44	4.44089209850063e-16\\
45	1.36002320516582e-15\\
46	1.57259307076027e-15\\
47	8.32667268468867e-17\\
48	1.11022302462516e-16\\
49	0\\
50	0\\
};

\addplot [color=mycolor2, line width=1.0pt, dashed]
table[row sep=crcr]{%
1	8.99104678423662\\
2	2.25694732588403\\
3	0.76323490357093\\
4	0.268185641756167\\
5	0.102994273803757\\
6	0.036512495806362\\
7	0.0141662585684595\\
8	0.00508871186911428\\
9	0.00197825580314322\\
10	0.000720193366498773\\
11	0.000280005343722838\\
12	0.000103156040760098\\
13	4.00760003271073e-05\\
14	1.49136007020456e-05\\
15	5.78709833112275e-06\\
16	2.17142134429983e-06\\
17	8.41483882702136e-07\\
18	3.17831859895884e-07\\
19	1.23007923693944e-07\\
20	4.67017727838731e-08\\
21	1.80533245542056e-08\\
22	6.88163012394609e-09\\
23	2.65753153537313e-09\\
24	1.01608698851773e-09\\
25	3.92069950511474e-10\\
26	1.50246568193618e-10\\
27	5.79378359245847e-11\\
28	2.22408416986917e-11\\
29	8.57223825214682e-12\\
30	3.29513990892752e-12\\
31	1.26867723976724e-12\\
32	4.8818692847501e-13\\
33	1.8790456795808e-13\\
34	7.14279423821752e-14\\
35	2.72859192989945e-14\\
36	1.11710661927834e-14\\
37	4.94391494818486e-15\\
38	2.28473718669871e-15\\
39	2.77035270246746e-15\\
40	2.28028536461812e-15\\
41	2.61653074715438e-15\\
42	2.22044604925031e-15\\
43	2.58834159649721e-15\\
44	2.22044604925031e-15\\
45	2.58834159649721e-15\\
46	2.22044604925031e-15\\
47	2.58834159649721e-15\\
48	2.22044604925031e-15\\
49	2.58834159649721e-15\\
50	2.22044604925031e-15\\
};

%\addlegendentry {s1}
\addplot[color=mycolor3, line width=1.0pt, dash dot]
table[row sep=crcr]{%
1	6.35718381667918\\
2	2.42434393481551\\
3	1.17335713519342\\
4	0.561347801643783\\
5	0.265704204821509\\
6	0.128861907765114\\
7	0.0629256923712213\\
8	0.0309579232759931\\
9	0.015320695355625\\
10	0.00761279998677837\\
11	0.00379251839776083\\
12	0.00189228919675016\\
13	0.000945026872650664\\
14	0.000472202622761171\\
15	0.000236015747499724\\
16	0.000117984514977891\\
17	5.89859253946143e-05\\
18	2.94912565000332e-05\\
19	1.47451708970579e-05\\
20	7.37246340700107e-06\\
21	3.68619926644621e-06\\
22	1.84309104229181e-06\\
23	9.21543253363966e-07\\
24	4.60771029680458e-07\\
25	2.30385357564613e-07\\
26	1.15192637892926e-07\\
27	5.7596307725718e-08\\
28	2.87981512426956e-08\\
29	1.43990744945999e-08\\
30	7.19953656146427e-09\\
31	3.59976847562339e-09\\
32	1.79988387617566e-09\\
33	8.99941614587271e-10\\
34	4.49970928820978e-10\\
35	2.24985687431792e-10\\
36	1.1249198068697e-10\\
37	5.6246433684188e-11\\
38	2.81230480286907e-11\\
39	1.40611710696371e-11\\
40	7.03013223116215e-12\\
41	3.51579017962878e-12\\
42	1.75795326830631e-12\\
43	8.78856803352496e-13\\
44	4.38792067788051e-13\\
45	2.19850473816317e-13\\
46	1.0970466228674e-13\\
47	5.47740524569476e-14\\
48	2.77118912878533e-14\\
49	1.37786397362254e-14\\
50	6.99057574353386e-15\\
};

%\addlegendentry {s2}

\end{semilogyaxis}
\end{tikzpicture}%
			\caption{Agreement error.}
	\end{subfigure}
\begin{subfigure}[b]{0.45\textwidth}
\definecolor{mycolor1}{rgb}{0.00000,0.44700,0.74100}%
\definecolor{mycolor2}{rgb}{0.85000,0.32500,0.09800}%
\definecolor{mycolor3}{rgb}{0.92900,0.69400,0.12500}%
\definecolor{mycolor4}{rgb}{0.63500,0.07800,0.18400}%
\begin{tikzpicture}[scale=0.9]

\begin{semilogyaxis}[%
width=1.6in,
height=1.4in,
at={(1.011in,0.642in)},
scale only axis,
xmin=0,
xmax=20,
xlabel = {$k$},
ymin=-20,
ymax=30,
ylabel = {$\sum_{i\in\mathcal{V}}||x^i(k)-\bar{x}(0)||$},
legend entries={{\small Algorithm~\ref{alg:resilient}}, {\small Algorithm~\ref{alg:switch}}, {\small Standard algorithm}, {\small W-MSR}},
legend style={legend cell align=left, align=left, draw=white!15!black},
legend pos={outer north east}
]
\addplot[color=mycolor4, line width=1.0pt, dotted]
table[row sep=crcr]{%
1	6.52835676044424\\
2	2.86272809590854\\
3	1.58934261428694\\
4	1.25429483939121\\
5	1.29690843595909\\
6	1.38490061460311\\
7	1.44159114989029\\
8	1.47436809559725\\
9	1.49209438531364\\
10	1.50137981642542\\
11	1.50617103685023\\
12	1.508623032513\\
13	1.50987291589115\\
14	1.51051066751758\\
15	1.51081695813369\\
16	1.5109678479628\\
17	1.51104297917764\\
18	1.51108005074739\\
19	1.51109842854022\\
20	1.51110735357622\\
21	1.51111176365053\\
22	1.51111388054742\\
23	1.51111492086424\\
24	1.51111541585429\\
25	1.51111564872774\\
26	1.51111575009339\\
27	1.51111579005633\\
28	1.51111580420875\\
29	1.5111158099576\\
30	1.5111158113564\\
31	1.51111581205071\\
32	1.5111158121407\\
33	1.51111581223348\\
34	1.51111581223066\\
35	1.51111581224066\\
36	1.51111581223639\\
37	1.51111581223685\\
38	1.51111581223552\\
39	1.51111581223527\\
40	1.5111158122349\\
41	1.51111581223466\\
42	1.51111581223458\\
43	1.51111581223453\\
44	1.51111581223451\\
45	1.5111158122345\\
46	1.5111158122345\\
47	1.5111158122345\\
48	1.5111158122345\\
49	1.5111158122345\\
50	1.5111158122345\\
};

\addplot [color=mycolor1, line width=1.0pt]
table[row sep=crcr]{%
1	6.52835676044424\\
2	1.70001502625727\\
3	0.595081474383987\\
4	0.230925395669012\\
5	0.0888699804999051\\
6	0.0344893655606191\\
7	0.0133269145886157\\
8	0.00515119652601112\\
9	0.00198815182382849\\
10	0.000766795236166358\\
11	0.00029545408193201\\
12	0.000113746744714303\\
13	4.37546027834973e-05\\
14	1.68177513096523e-05\\
15	6.45933032657267e-06\\
16	2.47914853100359e-06\\
17	9.50893554373175e-07\\
18	3.64497445930528e-07\\
19	1.39639526978166e-07\\
20	5.3467618150653e-08\\
21	2.04624778249006e-08\\
22	7.82755040478699e-09\\
23	2.99301041728049e-09\\
24	1.14398183250641e-09\\
25	4.37089621529681e-10\\
26	1.66945388479178e-10\\
27	6.37440775050554e-11\\
28	2.4331920426083e-11\\
29	9.28589825265209e-12\\
30	3.5428894368276e-12\\
31	1.3513714191137e-12\\
32	5.15846308233749e-13\\
33	1.96916014145106e-13\\
34	7.58160670609961e-14\\
35	2.92761501245397e-14\\
36	1.17380024981305e-14\\
37	5.47079355579994e-15\\
38	4.04170299099794e-15\\
39	2.973897742887e-15\\
40	2.88272997129038e-15\\
41	2.86906250128822e-15\\
42	2.88272997129038e-15\\
43	2.89639744129254e-15\\
44	2.89639744129254e-15\\
45	2.89639744129254e-15\\
46	2.89639744129254e-15\\
47	2.89639744129254e-15\\
48	2.89639744129254e-15\\
49	2.89639744129254e-15\\
50	2.89639744129254e-15\\
};

\addplot [color=mycolor2, line width=1.0pt, dashed]
table[row sep=crcr]{%
1	6.52835676044424\\
2	1.70001502625727\\
3	0.595081474383987\\
4	0.230925395669012\\
5	0.0888699804999058\\
6	0.0344893655606193\\
7	0.0133269145886158\\
8	0.00515119652601095\\
9	0.00198815182382873\\
10	0.000766795236166332\\
11	0.000295454081931927\\
12	0.000113746744713706\\
13	4.37546027831419e-05\\
14	1.68177513095607e-05\\
15	6.45933032657267e-06\\
16	2.47914853110579e-06\\
17	9.50893554524473e-07\\
18	3.64497445740654e-07\\
19	1.39639526797371e-07\\
20	5.34676175289798e-08\\
21	2.04624772491625e-08\\
22	7.82755017783203e-09\\
23	2.993010485665e-09\\
24	1.14398049996046e-09\\
25	4.37088969102616e-10\\
26	1.66945368705819e-10\\
27	6.37444867412786e-11\\
28	2.43321497191473e-11\\
29	9.28579919674697e-12\\
30	3.54266316637239e-12\\
31	1.35145836725693e-12\\
32	5.15077471127539e-13\\
33	1.97110886268318e-13\\
34	7.57234310776866e-14\\
35	2.86246373065054e-14\\
36	1.14838897024985e-14\\
37	4.52840076414809e-15\\
38	2.52908359533844e-15\\
39	2.55641853534276e-15\\
40	2.48414638367747e-15\\
41	2.4590637028824e-15\\
42	2.4590637028824e-15\\
43	2.4590637028824e-15\\
44	2.4590637028824e-15\\
45	2.4590637028824e-15\\
46	2.4590637028824e-15\\
47	2.4590637028824e-15\\
48	2.4590637028824e-15\\
49	2.4590637028824e-15\\
50	2.4590637028824e-15\\
};

%\addlegendentry {s1}
\addplot[color=mycolor3, line width=1.0pt, dash dot]
table[row sep=crcr]{%
1	6.52835676044424\\
2	2.47216745789169\\
3	1.1703634888813\\
4	0.628393844944701\\
5	0.399320635119882\\
6	0.378019490010352\\
7	0.395041274886357\\
8	0.409856509116041\\
9	0.418620481034159\\
10	0.423199152509131\\
11	0.425517603721295\\
12	0.426681091346566\\
13	0.427263357753198\\
14	0.427554509513192\\
15	0.427700064787574\\
16	0.427772831298481\\
17	0.427809210320019\\
18	0.427827398414357\\
19	0.42783649201857\\
20	0.427841038687806\\
21	0.427843311983683\\
22	0.427844448620557\\
23	0.427845016935882\\
24	0.427845301092681\\
25	0.427845443170843\\
26	0.427845514209859\\
27	0.42784554972935\\
28	0.427845567489091\\
29	0.42784557636896\\
30	0.427845580808894\\
31	0.427845583028861\\
32	0.427845584138845\\
33	0.427845584693836\\
34	0.427845584971332\\
35	0.42784558511008\\
36	0.427845585179455\\
37	0.427845585214141\\
38	0.427845585231485\\
39	0.427845585240157\\
40	0.427845585244493\\
41	0.427845585246661\\
42	0.427845585247744\\
43	0.427845585248286\\
44	0.427845585248557\\
45	0.427845585248693\\
46	0.427845585248761\\
47	0.427845585248794\\
48	0.427845585248811\\
49	0.42784558524882\\
50	0.427845585248824\\
};

%\addlegendentry {s2}
\end{semilogyaxis}
\end{tikzpicture}%
	\caption{Error to the average of initial states.}
\end{subfigure}
\caption{\modify{Performance comparison of different algorithms in the absence of faulty agent, where $\bar{x}(k) \triangleq \frac{1}{N}\sum_{j\in\mathcal{V}} x^j(k).$}}
\label{fig:woattack}
\end{figure}

\begin{figure}[!htbp]
	\centering
\definecolor{mycolor1}{rgb}{0.00000,0.44700,0.74100}%
\definecolor{mycolor2}{rgb}{0.85000,0.32500,0.09800}%
\definecolor{mycolor3}{rgb}{0.92900,0.69400,0.12500}%
\definecolor{mycolor4}{rgb}{0.49400,0.18400,0.55600}%
\begin{tikzpicture}[scale=0.9]

\begin{semilogyaxis}[%
width=1.6in,
height=1.4in,
at={(1.011in,0.642in)},
scale only axis,
xmin=0,
xmax=20,
xlabel = {$k$},
ymin=-20,
ymax=30,
ylabel = {$\sum_{i\in\mathcal{B}}||x^i(k)-\bar{x}(k)||$},
legend entries={{\small Algorithm~\ref{alg:resilient}}, {\small Algorithm~\ref{alg:switch}}, {\small Standard algorithm}, {\small W-MSR}},
legend style={legend cell align=left, align=left, draw=white!15!black},
legend pos={outer north east}
]
\addplot [color=mycolor4, line width=1.0pt,dotted]
table[row sep=crcr]{%
1	6.35718381667918\\
2	2.90253709429376\\
3	1.36052015245573\\
4	0.722957869497717\\
5	0.381223712313243\\
6	0.195951693940609\\
7	0.0963930469364378\\
8	0.0460119493439514\\
9	0.0216774960499969\\
10	0.0102793191907446\\
11	0.00509597177558195\\
12	0.00258319729509384\\
13	0.00129563747431363\\
14	0.000655382663227683\\
15	0.000279836772594724\\
16	0.000119670276768392\\
17	5.60506029146397e-05\\
18	2.66614970219634e-05\\
19	1.14536607026548e-05\\
20	5.07172293862533e-06\\
21	2.25413685915804e-06\\
22	1.00803221248919e-06\\
23	4.95051497323074e-07\\
24	2.52033546781178e-07\\
25	1.14921380567472e-07\\
26	6.54710856772556e-08\\
27	2.84848869116449e-08\\
28	1.52906937324955e-08\\
29	6.77847021116888e-09\\
30	4.19939131290465e-09\\
31	2.72904978231213e-09\\
32	1.9411803824512e-09\\
33	1.75122209538856e-09\\
34	1.41823913869471e-09\\
35	1.92802381037768e-09\\
36	1.52909376996733e-09\\
37	1.9212452701873e-09\\
38	2.02103828108744e-09\\
39	1.91785560004399e-09\\
40	1.82374330750878e-09\\
41	1.73777175147895e-09\\
42	1.7235928717258e-09\\
43	1.78127013617624e-09\\
44	1.8569213609009e-09\\
45	1.92034625833335e-09\\
46	1.97116312313658e-09\\
47	2.41903893752983e-09\\
48	1.86332126080125e-09\\
49	1.76768908480789e-09\\
50	1.0548595023098e-09\\
};

\addplot [color=mycolor1, line width=1.0pt]
table[row sep=crcr]{%
1	6.35718381667918\\
2	1.65038566772711\\
3	0.817104127010597\\
4	0.41040657892003\\
5	0.206700286396759\\
6	0.104700928769154\\
7	0.0508014815014124\\
8	0.0259503722354146\\
9	0.012513500738931\\
10	0.00643816117166044\\
11	0.00311701290030871\\
12	0.00160652929821038\\
13	0.000793417423292398\\
14	0.00040167372683285\\
15	0.000211139410041565\\
16	0.000103921263204809\\
17	5.58057099445485e-05\\
18	2.85296145816914e-05\\
19	1.43745303070924e-05\\
20	7.13589980988083e-06\\
21	3.44404195364757e-06\\
22	1.6913521996981e-06\\
23	8.33041217038669e-07\\
24	4.18593824083596e-07\\
25	2.05273928376414e-07\\
26	9.91790152544919e-08\\
27	4.94480557899667e-08\\
28	2.64771293821422e-08\\
29	1.40562175625647e-08\\
30	7.4144455659767e-09\\
31	3.93874977682747e-09\\
32	3.04503277884326e-09\\
33	1.70829193203617e-09\\
34	1.54786712084416e-09\\
35	1.38773114246673e-09\\
36	1.08226243940863e-09\\
37	1.29364429199119e-09\\
38	1.82960034880429e-09\\
39	2.00694938802532e-09\\
40	1.94731018947576e-09\\
41	1.91582945911982e-09\\
42	2.00101766215416e-09\\
43	2.14329667936661e-09\\
44	2.22406106956886e-09\\
45	2.18598675123492e-09\\
46	2.14990313430767e-09\\
47	2.55448433312293e-09\\
48	2.30303301873014e-09\\
49	2.15401611537866e-09\\
50	1.23573441549393e-09\\
};

\addplot [color=mycolor2, line width=1.0pt, dashed]
table[row sep=crcr]{%
1	6.35718381667918\\
2	1.39390639138735\\
3	0.379169585969789\\
4	0.279725743273293\\
5	0.271701006706112\\
6	0.261308660907912\\
7	0.246230806927192\\
8	0.224141758796896\\
9	0.193336225620132\\
10	0.153368754829082\\
11	0.105008104013537\\
12	0.0499767679535515\\
13	0.0139177267033263\\
14	0.0729892659054955\\
15	0.134691177614443\\
16	0.193476985596444\\
17	0.246769961409121\\
18	0.29227534721132\\
19	0.328035872191955\\
20	0.352506708897677\\
21	0.364618367668793\\
22	0.363823024378445\\
23	0.350124239372612\\
24	0.324094008834402\\
25	0.286890054645624\\
26	0.240314318492183\\
27	0.187066389107514\\
28	0.131943480296148\\
29	0.0885551341378436\\
30	0.0932951473983063\\
31	0.140243019975157\\
32	0.195149014591729\\
33	0.246805726593848\\
34	0.290951842624745\\
35	0.325099753529901\\
36	0.347600062915091\\
37	0.357440971139226\\
38	0.354205294995071\\
39	0.338064520939857\\
40	0.309783486297856\\
41	0.270746038872042\\
42	0.223067492058693\\
43	0.170070302199042\\
44	0.118535975524968\\
45	0.0891182207157962\\
46	0.111600030811892\\
47	0.161769266750252\\
48	0.215305250028743\\
49	0.264297337297747\\
50	0.305119278934343\\
};

\addplot[color=mycolor3, line width=1.0pt, dash dot]
table[row sep=crcr]{%
1	9.34702963534944\\
2	4.35013604515018\\
3	2.12566333359546\\
4	1.08914678058801\\
5	0.561291561660367\\
6	0.286634232297884\\
7	0.145236632721887\\
8	0.0732005003864608\\
9	0.036770728411746\\
10	0.0184340898703906\\
11	0.00923072951915613\\
12	0.00461915764895945\\
13	0.00231061931237306\\
14	0.00115559273738277\\
15	0.000577872862125216\\
16	0.000256082246036108\\
17	0.000137267852126007\\
18	7.1978171692766e-05\\
19	3.69448349660049e-05\\
20	1.8725937336567e-05\\
21	9.42814777115322e-06\\
22	4.73059203201444e-06\\
23	2.36945337719705e-06\\
24	1.18576950116871e-06\\
25	5.93145886825617e-07\\
26	2.9663828190511e-07\\
27	1.48335482256279e-07\\
28	7.41718269561802e-08\\
29	3.70869357839991e-08\\
30	1.85437230070504e-08\\
31	9.27192583670468e-09\\
32	4.63597892675102e-09\\
33	2.31799288596147e-09\\
34	1.13737309377385e-09\\
35	5.64705809773096e-10\\
36	2.81554822462564e-10\\
37	1.40604352228293e-10\\
38	7.02619807577562e-11\\
39	3.51216820471872e-11\\
40	1.75587861716604e-11\\
41	8.77928863308337e-12\\
42	4.38893800044864e-12\\
43	2.1947424760517e-12\\
44	1.09684437322261e-12\\
45	5.48826829250281e-13\\
46	2.74154717956994e-13\\
47	1.38491787763567e-13\\
48	6.98854595052527e-14\\
49	3.46166371556879e-14\\
50	1.7793408550121e-14\\
};
%\addlegendentry {s2}

\end{semilogyaxis}
\end{tikzpicture}%
	\caption{\modify{Performance comparison of different algorithms in the presence of faulty agent, where $\bar{x}(k) \triangleq \frac{1}{|\mathcal{B}|}\sum_{j\in\mathcal{B}} x^j(k).$}}
	\label{fig:wattack}
\end{figure}

% \begin{figure}[!htbp]
%	\centering
%	\input{tikz/withattack.tikz}
%	\caption{\revise{The trajectory of local states in the presence of faulty agent, the behavior of which is denoted by the red line. The left hand and right hand shows the performance of Algorithm~\ref{alg:switch} and that in \cite{olfati2004consensus}, respectively.}}
%	\label{fig:withattack}
%\end{figure}
%
% \begin{figure}[!htbp]
%	\centering
%	\input{tikz/withoutattack.tikz}
%	\caption{\revise{The trajectory of local states in the absence of faulty agent. The left hand and right hand shows the performance of Algorithm~\ref{alg:switch} and that in \cite{olfati2004consensus}, respectively.}}
%	\label{fig:withoutattack}
%\end{figure}

\subsection{A $3$-dimensional example}
To further evaluate the computational complexity, we perform Algorithm~\ref{alg:resilient} in a $3$D system. It is assumed that $8$ agents cooperate through a complete graph with one of the agents being faulty. According to \cite{leblanc2013resilient}, this network is $(4,2)$-robust. Therefore, it is able to tolerate this single misbehaving agent as claimed in Theorem~\ref{thm:converge} and also verified in Fig.~\ref{fig:3d}. On a Macbook Pro with a 2.6 GHz processor, at every time step, each agent takes around $2$ seconds to make an update by calculating the auxiliary point using CVX (\cite{grant2014cvx}). As analyzed in Sections~\ref{sec:middle_points} and \ref{sec:compare}, the existing algorithms will inevitably introduce higher cost in this complete graph where each agent has a large neighborhood, since the calculation therein increases exponentially with the number of in-neighbors. 

It is also noted that our computation will be super fast if there are no faulty nodes, since Algorithm~\ref{alg:switch} performs exactly as a standard average consensus algorithm in this scenario.

\begin{figure}[!htbp]
	\centering
	\input{tikz/3dnew.tikz}
	\caption{\revise{The trajectory of local states under Algorithm~\ref{alg:resilient}, where the area within the blue polyhedron is the convex hull formed by the initial states of benign agents.}}
	\label{fig:3d}
\end{figure}

%%%%%%%%%%%%%%%%%%%%%%%%%%%%%%%%%%%%%%%%%%%%%%%%%%%%%%%%%%%%%%%%%%%%%%%%%%%%%%%%%%%%%%
\section{Conclusion}\label{sec:conclusion}
Due to its wide applications, the distributed coordination in networked systems has attracted much research interest. In this paper, we are interested in the achievement of resilient vector consensus under misbehaving agents. We propose a general framework of consensus algorithms and provide verifiable conditions for checking the resiliency of it. The lower bound on convergence rate is also given. Based on the derived conditions, two specific algorithms are then developed. As the consensus arguably forms the foundation for distributed computing, the results in this paper provides solid foundations for future works to develop resilient coordination protocols in other consensus-based problems.

%%%%%%%%%%%%%%%%%%%%%%%%%%%%%%%%%%%%%%%%%%%%%%%%%%%%%%%%%%%%%%%%%%%%%%%%%%%%%%%%%%%%%%

\bibliographystyle{agsm} 
\bibliography{reference} 

\appendix
\section{Proof of Theorem \ref{thm:resilientCon}}\label{app:A}
In order to show that benign agents exponentially achieve the resilient consensus, we shall respectively prove the $\delta$-validity and agreement conditions.

\textbf{\revise{Part I: Achievement of $\delta$-validity condition}} 

The $\delta$-validity condition is guaranteed by the following results: 
\begin{lemma}\label{lmm:H_i}
	Under the conditions of Theorem \ref{thm:resilientCon}, it holds at any time $k$ that: 
	\begin{equation}
	\Xi(k+1) \subset \Xi(k) + \mathcal{B}_{\delta(k)},
	\end{equation}
	where $\Xi(k)$ is defined in \eqref{eqn:H_i}.
\end{lemma}
\vspace{-10pt}
\begin{pf*}{Proof.}
	Notice that $\Xi(k)$ is convex at any time. Given C1 and C3, the below statement is clear for any benign agent $i\in\mathcal{B}$:
	\begin{equation}\label{eqn:H}
	x^i(k+1) \in \{x:||x-\chi|| \leq \delta(k), \exists \chi\in \Xi(k)\},
	\end{equation}
	where $\delta(k)$ is defined in \eqref{eqn:deltadef}. The proof is completed. 
	\hfill$\square$
\end{pf*}

\begin{proposition}\label{thm:validity}
	Under the conditions of Theorem \ref{thm:resilientCon}, the $\delta$-validity condition is guaranteed by \eqref{eqn:update_benign}.
\end{proposition}
\begin{pf*}{Proof.}
	It follows that 
	\begin{equation}
	\sum_{k=0}^{\infty} \delta(k) <\sum_{i\in\mathcal{B}} \Big(\sum_{k=0}^{\infty} ||\varepsilon^i(k)||\Big)<\infty, 
	\end{equation}
	where the last inequality holds by \eqref{eqn:varepsilon}. The proof is thus trivial by evoking Lemma~\ref{lmm:H_i}. 
	\hfill$\square$
\end{pf*}

\textbf{\revise{Part II: Achievement of agreement condition}} 

In order to establish the agreement condition, it is equivalent to show that benign agents reach a consensus at any dimension. Due to the symmetry between different dimensions, without loss of generality, we would only focus on the first entry of local states. 
To this end, we first prove \eqref{eqn:delta}, namely:
\begin{lemma}\label{lmm:gap}
	Under the conditions of Theorem \ref{thm:resilientCon}, it holds at any $k\in\mathbb{Z}_{\geq 0}$ that:
	\begin{equation*}
	\Delta_1(k+|\mathcal{B}|)\leq \bigg(1-\frac{(0.5\alpha)^{|\mathcal{B}|}}{2}\bigg)\Delta_1(k)+2\sum_{\tau=k}^{k+|\mathcal{B}|-1}\delta(\tau).
	\end{equation*}
\end{lemma}
\begin{pf*}{Proof.}
	As a direct result of Lemma~\ref{lmm:H_i}, it yields for any $\ell\in\mathcal{D}$ that
	\begin{equation}\label{eqn:M}
	\begin{split}
	M^\mathcal B_\ell(k+1)&\leq M^\mathcal B_\ell(k)+\delta(k),\\
	m^\mathcal B_\ell(k+1)&\geq m^\mathcal B_\ell(k)-\delta(k),
	\end{split}
	\end{equation}
	where $M^\mathcal B_\ell(k)$ and $m^\mathcal B_\ell(k)$ are given in \eqref{eqn:M(k),m(k)}.
	
	Then suppose that at time $k$, it holds that $M^\mathcal B_1(k)\neq m^\mathcal B_1(k)$, i.e., $\Delta_1(k)>0$. Let us define $\epsilon_0 = \Delta_1(k)/2$. It is trivial to conclude that $\mathcal{B}_1^M(k,M^\mathcal B_1(k)-\epsilon_0)$ and $\mathcal{B}_1^m(k,m^\mathcal B_1(k)+\epsilon_0)$ are both disjoint and nonempty. Hence, a benign agent exists, labeled as $j$, such that either Condition 4a) or 4b) holds.
	
	Without loss of generality, let $j\in \mathcal{B}_1^M(k,M^\mathcal B_1(k)-\epsilon_0)$ be such an agent. Therefore, there exists a point in $\Lambda^j(k)$, the first entry of which is upper bounded by $M^\mathcal B_1(k)-\epsilon_0$. We hence conclude that $m_1(\Lambda^j(k))\leq M^\mathcal B_1(k)-\epsilon_0$. In view of C3, it follows that
	$M_1(\Lambda^j(k)) \leq M^\mathcal B_1(k).$
	Since $\tilde{x}^j(k)=\cen(\Lambda^j(k))$, we obtain that 
	\begin{equation}
	\begin{split}
	\tilde{x}_1^j(k) &\leq 0.5 M^\mathcal B_1(k) + 0.5(M^\mathcal B_1(k)-\epsilon_0)\\
	& = M^\mathcal B_1(k) -0.5 \epsilon_0.
	\end{split}
	\end{equation}
	Consequently, it follows that
	\begin{equation}\label{eqn:upperbound}
	\begin{split}
	x^j_1(k+1) &=a^j x_1^j(k) + (1-a^j) \tilde{x}_1^j(k) + \varepsilon_1^j(k)\\
	&\leq (1-\alpha) M^\mathcal B_1(k) + \alpha[M^\mathcal B_1(k) - 0.5 \epsilon_0]+ \delta(k)\\&=M^\mathcal B_1(k)-0.5\alpha\epsilon_0+ \delta(k).
	\end{split}
	\end{equation}
	Moreover, for any benign node $i \in \mathcal{V}\backslash\mathcal{V}_1^M(k, M^\mathcal B_1(k)-\epsilon_0)$, it holds that
	\begin{equation*}
	\begin{split}
	x^i_1(k+1) &=a^i x_1^i(k) + (1-a^i) \tilde{x}_1^i(k) + \varepsilon_1^i(k)\\
	&\leq \alpha [M^\mathcal B_1(k)-\epsilon_0] + (1-\alpha) M^\mathcal B_1(k) + \delta(k)\\&=M^\mathcal B_1(k)-\alpha\epsilon_0+ \delta(k).
	\end{split}
	\end{equation*}
	Therefore, the upper bound in \eqref{eqn:upperbound} also applies to any benign agent in $\mathcal{V}\backslash\mathcal{V}_1^M(k,M^\mathcal B_1(k)-\epsilon_0)$.
	
	Similarly, if $j\in \mathcal{B}_1^m(k,m^\mathcal B_1(k)+\epsilon_0)$,  and $\Lambda^j(k)$ contains some point with first entry lower bounded by $\epsilon_2$, we obtain that $\tilde{x}_1^j(k) \geq m^\mathcal B_1(k) + 0.5 \epsilon_0$ and have an analogous result that $x^j_1(k+1)\geq m^\mathcal B_1(k)+0.5\alpha\epsilon_0- \delta(k),$ which again, is the lower bound for every benign agent in $\mathcal{V}\backslash\mathcal{V}_1^m(k,m^\mathcal B_1(k)+\epsilon_0)$.
	
	\revise{Define $\epsilon_1=0.5\alpha\epsilon_0-\delta(k)<\epsilon_0$. Let us consider the sets $\mathcal{B}_1^M(k+1,M^\mathcal B_1(k)-\epsilon_1)$ and $\mathcal{B}_1^m(k+1,m^\mathcal B_1(k)+\epsilon_1)$. From former discussions, we know that one of the following statements must be true:
		\begin{enumerate}
			\item At least one benign agent in $\mathcal{B}_1^M(k,M^\mathcal B_1(k)-\epsilon_0)$ has its first component decreasing to below (or equal) $M^\mathcal B_1(k)-\epsilon_1$, which leads to $\mathcal{B}_1^M(k+1,M^\mathcal B_1(k)-\epsilon_1)\subsetneq\mathcal{B}_1^M(k,M^\mathcal B_1(k)-\epsilon_0)$;
			\item At least one benign agent in $\mathcal{B}_1^m(k,m^\mathcal B_1(k)+\epsilon_0)$ has its first component increasing to above (or equal) $m^\mathcal B_1(k)+\epsilon_1$, which leads to $\mathcal{B}_1^m(k+1,m^\mathcal B_1(k)+\epsilon_1)\subsetneq\mathcal{B}_1^m(k,m^\mathcal B_1(k)+\epsilon_0)$.
		\end{enumerate}
		We therefore conclude that $|\mathcal{B}_1^M(k+1,M^\mathcal B_1(k)-\epsilon_1)|+|\mathcal{B}_1^m(k+1,m^\mathcal B_1(k)+\epsilon_1)|<|\mathcal{B}_1^M(k,M^\mathcal B_1(k)-\epsilon_0)|+|\mathcal{B}_1^m(k,m^\mathcal B_1(k)+\epsilon_0)|.$}
	
	Then let us consider the update at $k+2$. Suppose that $\mathcal{B}_1^M(k+1,M^\mathcal B_1(k)-\epsilon_1)$ and $\mathcal{B}_1^m(k+1,m^\mathcal B_1(k)+\epsilon_1)$ are still disjoint. Again, there exists an agent $l$ such that either 4a) or 4b) holds. Suppose 4a) is true. Similarly, one has that
	\begin{equation}
	\begin{split}
	\tilde{x}^l_1(k+1) &\leq 0.5 M^\mathcal B_1(k+1)+0.5(M^\mathcal B_1(k)-\epsilon_1)\\
	& \leq 0.5(M^\mathcal B_1(k)+\delta(k))+0.5(M^\mathcal B_1(k)-\epsilon_1)\\
	&= M^\mathcal B_1(k)+0.5 \delta(k) - 0.5\epsilon_1.
	\end{split}
	\end{equation}
	Combining it with \eqref{eqn:M} yields that
	\begin{equation}
	\begin{split}
	x_1^l(k+2) &=a^l x_1^l(k+1) + (1-a^l) \tilde{x}_1^l(k+1) + \varepsilon_1^l(k+1)\\
	&\leq (1-\alpha)[M^\mathcal B_1(k)+\delta(k)]+ \alpha [M^\mathcal B_1(k)+0.5 \delta(k) \\&\qquad- 0.5\epsilon_1] + \delta(k+1)\\
	&=M^\mathcal B_1(k)-(0.5\alpha)^2\epsilon_0+ \delta(k)+\delta(k+1)\\
	&=M^\mathcal B_1(k)-\epsilon_2,
	\end{split}
	\end{equation}
	where $\epsilon_2=(0.5\alpha)^2\epsilon_0- \delta(k)-\delta(k+1)$. The above upper bound also holds for any benign node $r\in\mathcal{V}\backslash\mathcal{V}_1^M(k+1,M^\mathcal B_1(k)-\epsilon_1)$. To see this, one considers the update of agent $r$, namely,
	\begin{equation*}
	\begin{split}
	x^r_1(k+2) &=a^r x_1^r(k+1) + (1-a^r) \tilde{x}_1^r(k+1) + \varepsilon_1^r(k+1)\\
	&\leq \alpha [M^\mathcal B_1(k)-\epsilon_1]  + (1-\alpha) [M^\mathcal B_1(k)+\delta(k)] \\&\qquad+ \delta(k+1)\\
	&=M^\mathcal B_1(k)-\alpha\epsilon_1+(1-\alpha) \delta(k)+\delta(k+1)\\
	&=M^\mathcal B_1(k)-0.5\alpha^2\epsilon_0+ \delta(k)+\delta(k+1)\\
	&\leq M^\mathcal B_1(k)-(0.5\alpha)^2\epsilon_0+ \delta(k)+\delta(k+1)\\
	&=M^\mathcal B_1(k)-\epsilon_2.
	\end{split}
	\end{equation*}
	
	On the other hand, if $l\in\mathcal{B}_1^m(k+1,m^\mathcal B_1(k)+\epsilon_1)$, then $\Lambda^l(k+1)$ must contain some point with first entry lowered bounded by $m^\mathcal B_1(k)+\epsilon_1$. Following the previous arguments gives that $x_1^l(k+2)\geq m^\mathcal B_1(k)+\epsilon_2$, which is also the lower bound for every benign node in $\mathcal{V}\backslash\mathcal{V}_1^m(k+1,m^\mathcal B_1(k)+\epsilon_1)$. Hence, it must be true that, either $\mathcal{B}_1^M(k+2, M^\mathcal B_1(k)-\epsilon_2)\subsetneq\mathcal{B}_1^M(k+1,M^\mathcal B_1(k)-\epsilon_1),$ or $\mathcal{B}_1^m(k+2,m^\mathcal B_1(k)+\epsilon_2)\subsetneq\mathcal{B}_1^m(k+1,m^\mathcal B_1(k)+\epsilon_1),$ or both.  
	
	Therefore, for any $t \geq 1$, we can define $\epsilon_t=(0.5\alpha)^t\epsilon_0- \sum_{\tau=k}^{k+t-1}\delta(\tau)$. As long as both $\mathcal{B}_1^M(k+t,M^\mathcal B_1(k)-\epsilon_{t})$ and $\mathcal{B}_1^m(k+t,m^\mathcal B_1(k)+\epsilon_{t})$ are nonempty, we can repeat the above analysis and conclude that at least one of these two sets will shrink at the next time step. \revise{That is, for any $t\in\mathbb{Z}_{\geq 0}$, either $|\mathcal{B}_1^M(k+t+1,M^\mathcal B_1(k)-\epsilon_{t+1})|<|\mathcal{B}_1^M(k+t,M^\mathcal B_1(k)-\epsilon_{t})|$, or $|\mathcal{B}_1^m(k+t+1,m^\mathcal B_1(k)+\epsilon_{t+1})|<|\mathcal{B}_1^m(k+t,m^\mathcal B_1(k)+\epsilon_{t})|$, or both. }
	
	\revise{Since both $|\mathcal{B}_1^M(k,M^\mathcal B_1(k)-\epsilon_0)|$ and $|\mathcal{B}_1^m(k,m^\mathcal B_1(k)+\epsilon_0)|$ are upper bounded by $|\mathcal{B}|$, one of these sets would be empty after $|\mathcal{B}|$ steps. Namely, one of following statements must be true:
		\begin{enumerate}
			\item $\mathcal{B}_1^M(k+|\mathcal{B}|,M^\mathcal B_1(k)-\epsilon_{|\mathcal{B}|})=\varnothing,$
			\item $\mathcal{B}_1^m(k+|\mathcal{B}|,m^\mathcal B_1(k)+\epsilon_{|\mathcal{B}|})=\varnothing$
		\end{enumerate}
		Without loss of generality, we assume that the first statement holds.} By \eqref{eqn:V^MandV^m}--\eqref{eqn:B^MandB^m}, at time $k+|\mathcal{B}|$, all the fault-free agents have their first elements upper bounded by $M^\mathcal B_1(k)-\epsilon_{|\mathcal{B}|}$, i.e., 
	$M^\mathcal B_1(k+|\mathcal{B}|)\leq M^\mathcal B_1(k)-\epsilon_{|\mathcal{B}|}.$ On the other hand, from \eqref{eqn:M}, we have
	$m^\mathcal B_1(k+|\mathcal{B}|)\geq m^\mathcal B_1(k)-\sum_{\tau=k}^{k+|\mathcal{B}|-1}\delta(\tau).$
	Therefore, one concludes that                                                                                           
	\begin{equation*}
	\begin{split}
	\Delta_1(k+|\mathcal{B}|)&\leq \Delta_1(k)-\epsilon_{|\mathcal{B}|}+\sum_{\tau=k}^{k+|\mathcal{B}|-1}\delta(\tau)\\&=
	\bigg(1-\frac{(0.5\alpha)^{|\mathcal{B}|}}{2}\bigg)\Delta_1(k)+2\sum_{\tau=k}^{k+|\mathcal{B}|-1}\delta(\tau).
	\end{split}
	\end{equation*}
	The proof is thus completed. \hfill$\square$
\end{pf*}

In order to show that $\Delta_1(k)$ asymptotically approaches $0$, we introduce the lemma below:
\revise{
	\begin{lemma}[\hspace{1pt}{\cite[Lemma~7]{nedic2010constrained}}]\label{lmm:nedic}
		Let $0<\beta<1$ and let $\left\{\gamma_{k}\right\}$ be a positive scalar sequence. Assume that $\lim _{k \rightarrow \infty} \gamma_{k}=0 .$ Then
		$$
		\lim _{k \rightarrow \infty} \sum_{\ell=0}^{k} \beta^{k-\ell} \gamma_{\ell}=0.
		$$
	\end{lemma}
}

We therefore conclude the following result on achieving the agreement condition of resilient consensus:

\begin{proposition}\label{thm:consensus}
	Under the conditions of Theorem \ref{thm:resilientCon}, the agreement condition is exponentially guaranteed by \eqref{eqn:update_benign}.
\end{proposition}
\begin{pf*}{Proof.}
	Again, we consider the first component of states. Let us fix time $k$. In view of \eqref{eqn:delta}, for any $t\in\mathbb{Z}_{\geq 0}$, it holds that
	\begin{equation}\label{eqn:error}
	\begin{split}
	\Delta_1&(k+t|\mathcal{B}|)\leq \bigg(1-\frac{(0.5\alpha)^{|\mathcal{B}|}}{2}\bigg)^t\Delta_1(k)\\&+2\sum_{l=0}^{t-1}\bigg(1-\frac{(0.5\alpha )^{|\mathcal{B}|}}{2}\bigg)^{t-1-l}\sum_{\tau=k+(t-1-l)|\mathcal{B}|}^{k+(t-l)|\mathcal{B}|-1}\delta(\tau).
	\end{split}
	\end{equation}
	Due to \eqref{eqn:varepsilon}, we conclude that $$\lim_{t\rightarrow \infty}\sum_{\tau=k+(t-1-l)|\mathcal{B}|}^{k+(t-l)|\mathcal{B}|-1}\delta(\tau)=0.$$ In view of Lemma~\ref{lmm:nedic}, the second term in RHS of \eqref{eqn:error} goes to $0$ as $t\to\infty$. Therefore, it follows that
	\begin{equation}
	\lim_{t\rightarrow \infty}\Delta_1(k+t|\mathcal{B}|)=0.
	\end{equation}
	Since the above equation holds true for each $k\in\mathbb{N}$, one concludes that
	\begin{equation}
	\lim\limits_{k\to\infty}	\Delta_1(k) = 0.
	\end{equation}
	Given the symmetry of different dimensions, the proof is thus completed.
	\hfill$\square$
\end{pf*}

\textbf{\revise{Part III: Conclusion on the proof of Theorem~\ref{thm:resilientCon}}} 

Combining Part I and Part II, we finally complete the proof of Theorem~\ref{thm:resilientCon}. Specifically, the statements in Theorem~\ref{thm:resilientCon} are verified as follows:

(1) In view of Propositions~\ref{thm:validity} and \ref{thm:consensus}, both the $\delta$-validity and agreement conditions are guaranteed by C1--C4. Recalling Definition~\ref{def:resilientCon}, we conclude that the benign agents exponentially achieve the resilient consensus. 

(2) The second statement has been proven in Lemma~\ref{lmm:gap}.

\section{Proof of Corollary \ref{lm:existence}}
Corollary~\ref{lm:existence} is obvious when $n=0$. Thus we only focus on the scenario when $n\geq 1$.

According to Definition~\ref{def: intersection}, $\varPsi (\mathcal{A}, n)$ is an intersection of $\binom{m}{n}$ convex hulls. Since $m\geq n(d+1)+1$, it is trivial to prove that $\binom{m}{n}>d$ holds.

On the other hand, each of these convex hulls is created by excluding $n$ elements of $\mathcal{A}$. Then consider any $d+1$ of them, they discard at most $n(d+1)$ points in all. Since $m\geq n(d+1)+1$, it must be the case that at least one point in $\mathcal{A}$ is retained by all of them. This indicates that any $d+1$ convex hulls must have a nonempty intersection.
By applying Helly's Theorem, the proof is completed.

\section{Proof of Lemma~\ref{lmm:convexhull}}
We shall prove the statements of Lemma~\ref{lmm:convexhull} in order.

(1) To prove the first statement, we need the following lemmas:

\begin{lemma}\label{pro:subset}
	Consider two collections of sets $\left\{A_i\right\}_{i\in I}$ and $\left\{B_j\right\}_{j\in \mathcal J}$. If for any $j\in \mathcal J$, there exists an $i^*\in \mathcal I$ such that $A_{i^*}\subset B_j$, then
	\begin{align*}
		\bigcap_{i\in\mathcal I}A_i \subset\bigcap _{j\in\mathcal J} B_j .
	\end{align*}
\end{lemma}
\begin{pf*}{Proof.}
	Denote a subset of $\left\{A_i\right\}_{i\in I}$ as $\left\{A_{i^*}\right\}_{i^*\in I}$, such that $\left\{A_{i^*}\right\}_{i^*\in I}$ contains all $A_{i^*}$ which has a superset in $\left\{B_j\right\}_{j\in \mathcal J}$.
	The proof is then completed by noticing that 
	\begin{align*}
		\bigcap_{i\in\mathcal I}A_i \subset \bigcap_{i^*\in\mathcal I}A_{i^*} \subset \bigcap _{j\in\mathcal J} B_j .
	\end{align*}
	\hfill$\square$
\end{pf*}

\begin{lemma}\label{lm: subset}
	Consider any set $\mathcal{A}_1$ with cardinality $m_1$ and $\mathcal{A}_2$ with cardinality $m_2$. If $\mathcal{A}_1 \subset \mathcal{A}_2$, then for any $n\leq m_1$, the following statement holds: $$\varPsi (\mathcal{A}_1, n) \subset \varPsi (\mathcal{A}_2, n).$$
\end{lemma}
\begin{pf*}{Proof.}
We show that every set $S_2$ in $\mathcal S(\mathcal A_2,n)$ is a superset for some set $S_1$ in $\mathcal S(\mathcal A_1,n)$. To see this, notice that
\begin{align*}
	S_2 &= \mathcal A_2 \backslash S_2^c \supset \left(\mathcal A_2\backslash S_2^c\right)\cap \mathcal A_1\\
	&=\mathcal A_1 \backslash \left(S_2^c\cap\mathcal A_1\right),
\end{align*}
where $S_2^c = \mathcal A_2\backslash S_2$ is a set with cardinality $n$. Notice that $S_2^c\cap\mathcal A_1$ has cardinality no greater than $n$, which means that $\mathcal A_1 \backslash \left(S_2^c\cap\mathcal A_1\right)$ is a superset of some set in $\mathcal S(\mathcal A_1,n)$. The proof is thus finished by invoking Lemma~\ref{pro:subset}.
\hfill$\square$
\end{pf*}

We are ready to prove the first statement. From Algorithm~\ref{alg:resilient}, it is noticed that $\mathcal{Y}^i(p,k), \mathcal{Z}^i(p,k) \subset \mathcal{X}^i(k)$ holds for any $i\in\mathcal{B}$ and $p\in\mathcal{D}$. Recalling the definition of $S^i(k)$, namely \eqref{def:S_i}, it is not difficult to conclude that $y^i(p,k), z^i(p,k) \in \mathcal S^i(k),$ as a direct result of Lemma~\ref{lm: subset}. 

(2) Consider the scenario under either $F$-local or $F$-total attack. For a benign agent $i$, it has no less than $|\mathcal{X}^i(k)|-F$ benign in-neighbors. By Definitions \ref{def: S} and \ref{def: intersection}, one obtains that $\mathcal{S}^i(k)$ is included in the convex hull formed by any $|\mathcal{X}^i(k)|-F$ in-neighboring values. Hence, it is trivial to derive that $\mathcal{S}^i(k)$ is a subset of the convex hull formed by the benign in-neighbors' states, that is, $\mathcal{S}^i(k)\subset\Xi(k).$

(3) Combining the first two statements, it is not difficult to conclude that
\begin{equation}
	\Lambda^i(k) \subset \Xi(k),
\end{equation}
where $\Lambda^i(k)$ is defined in \eqref{eqn:Lambda^i}. In view of \eqref{eqn:auxiliary1}, one concludes that C3 holds.

\section{Proof of Lemma \ref{lmm:topo1}}
For proving Lemma \ref{lmm:topo1}, let us introduce the lemma below:
\begin{lemma}\label{lm:y^i}
	Let $\mathcal{A}$ be a set with $|\mathcal{A}|= (d+1)F+1$. The following relations hold for any linear function $l(x)$:
	\begin{enumerate}
		\item If there exist at least $dF+1$ points $\bar{x}$ in $\mathcal{A}$ such that $l(\bar{x})\leq m$, then for any point $y \in \varPsi(\mathcal{A},F)$, $l(y) \leq m$ holds;
		\item If there exist at least $dF+1$ points $\bar{x}$ in $\mathcal{A}$ such that $l(\bar{x})\geq M$, then for any point $z \in \varPsi(\mathcal{A},F)$, $l(z) \geq M$ holds.
	\end{enumerate}
\end{lemma}
\begin{pf*}{Proof.}
	By Corollary \ref{lm:existence}, $\varPsi (\mathcal{A}, n)\neq \varnothing$. We then show the rationale of the statements as follows:
	\begin{enumerate}
		\item Let us consider the convex hull, denoted by $\Psi$, formed by any $dF+1$ points $\bar{x}$ such that $l(\bar{x})\leq m$ and $\bar x \in\mathcal{A}$. We could infer from Definition \ref{def: intersection} that $\varPsi (\mathcal{A}, F)$ is a subset of the convex hull formed by any  $dF+1$ points in $\mathcal{A}$. Therefore, $\varPsi (\mathcal{A}, F)\subseteq \Psi$.
		Clearly, for any element $x$ in $\Psi$, it also follows that $l(x)\leq m$. Therefore, the first statement holds.
		\item The second statement is proved in a similar manner as above. \hfill$\square$
	\end{enumerate}
\end{pf*}

Given any $x\in\mathbb{R}^d$, let us consider the following function: $$l_\ell(x)= e_\ell^T x,$$ where $e_\ell$ is the $\ell$-th canonical basis vector in $\mathbb{R}^d$. Hence $l_\ell(x)$ returns the $\ell$-th entry of $x$. Then among all $l_\ell(x)$'s, where $x \in \mathcal{X}^i(k)$, we respectively denote by $\overline{m}^i_\ell(k)$ and $\overline{M}^i_\ell(k)$ the $(dF+1)$-th smallest and $(dF+1)$-th largest values.
From Lemma \ref{lm:y^i}, one immediately has the following results: 
\begin{equation}
	\begin{aligned}
		y^i_\ell(p,k) &\leq \overline{m}^i_\ell(k),\\
		z^i_\ell(p,k) &\geq \overline{M}^i_\ell(k),
	\end{aligned}
\end{equation}
where \revise{$y^i_\ell(p,k)$ [resp. $z^i_\ell(p,k)$] refers to the $\ell$-th entry of $y^i(p,k)$ [resp. $z^i(p,k)$]}. Namely, for any $\ell\in\mathcal{D}$, the $\ell$-th entry of $y^i(p,k)$ is upper bounded by $\overline{m}^i_\ell(k)$. Similarly, the $\ell$-th entry of $z^i(p,k)$ is lower bounded by $\overline{M}^i_\ell(k)$.

Given \eqref{eqn:Lambda^i}, for any $\ell\in\mathcal{D}$, we conclude that
\begin{equation}
	\begin{split}
		m_\ell(\Lambda^i(k)) &= \min_{p \in\mathcal{D}}\big(\min \{y_\ell^i(p,k), z_\ell^i(p,k)\}\big),\\
		M_\ell(\Lambda^i(k)) &= \max_{p \in\mathcal{D}}\big(\max\{y_\ell^i(p,k), z_\ell^i(p,k)\}\big),
	\end{split}
\end{equation}
where $m_\ell(\Lambda^i(k))$ and $M_\ell(\Lambda^i(k))$ are defined in Section~\ref{sec:notation}. Therefore, it holds that 
\begin{equation}\label{eqn:widehatbound}
	\begin{split}
		m_\ell(\Lambda^i(k)) & \leq \overline{m}^i_\ell(k),\\
		M_\ell(\Lambda^i(k)) &\geq \overline{M}^i_\ell(k).
	\end{split}
\end{equation}

With the preparations above, we finally provide the detailed proof of Lemma~\ref{lmm:topo1} for the case under statement 1), with which the case under 2) can be proved similarly.
	
To this end, consider any $\ell \in \mathcal{D}$ and any $\epsilon_1,\epsilon_2 \in \mathbb{R}$. Suppose that $\mathcal{B}_\ell^M(k,\epsilon_1)$ and $\mathcal{B}_\ell^m(k,\epsilon_2)$ are disjoint and nonempty. \revise{As the network is $((d+1)F+1)$-robust, there must exist one benign node, labeled as $j$,} in either $\mathcal{B}_\ell^M(k,\epsilon_1)$ or $\mathcal{B}_\ell^m(k,\epsilon_1)$ that has at least $(d+1)F+1$ in-neighbors outside its set.   
	
	Without loss of generality, let $j\in \mathcal{B}_\ell^M(k,\epsilon_1)$ be such an agent. It has no less than $(d+1)F+1$ neighbors in $\mathcal{V}\backslash\mathcal{B}_\ell^M(k,\epsilon_1)$. Moreover, under the $F$-local attack model, no less than $dF+1$ points in agent $j$'s in-neighborhood have the $\ell$-th entries upper bounded by $\epsilon_1$. Therefore, one has that $\overline{m}^j_\ell(k)\leq \epsilon_1$. Invoking \eqref{eqn:widehatbound}, $m_\ell(\Lambda^j(k))\leq \overline{m}^j_\ell(k)\leq \epsilon_1$. Hence, $\Lambda^j(k)$ must contain a point, the $\ell$-th entry of which is upper bounded by $\epsilon_1$. Similarly, if agent $j\in \mathcal{B}_\ell^m(k,\epsilon_2)$, $\Lambda^j(k)$ must contain some point with $\ell$-th entry lower bounded by $\epsilon_2$. The proof is thus completed.

\section{Proof of Lemma~\ref{lmm:equvi}}

	Notice that agent $i$ sets $\breve{x}^i(k)$ as $\bar{x}^i(k)$ if and only if $||x^j(k)-\bar{x}^i(k)||\leq \lambda(k)$ holds for any $ j\in\mathcal{N}_i^+$. As proved in Section~\ref{sec:alg}, $\cen(\Lambda^i(k))$ is a convex combination of these in-neighboring states. By using triangle inequality, it hence follows that
	\begin{align}
		\big| \big|\cen(\Lambda^i(k))-\bar{x}^i(k)\big| \big|\leq \lambda(k).
	\end{align}
	We thus rewrite \eqref{eqn:ave}--\eqref{eqn:resilient} as
	\begin{equation}
		\breve{x}^i(k) = \cen(\Lambda^i(k))+ \xi^i(k),
	\end{equation}
	where $||\xi^i(k)||\leq \lambda(k)$. Setting $\varepsilon^i(k)=(1-a^{ii})\xi^i(k)$ yields the conclusion.

%\subsection{Proof of Proposition \ref{thm:consensus_total}}
%Proposition \ref{thm:consensus_total} is proved in a similar manner to that of Proposition \ref{thm:consensus_local}. The essential point is that if $\mathcal{V}^M(k,k+\kappa d,\epsilon_{\kappa d})$ and $\mathcal{V}^m(k,k+\kappa d,\epsilon_{\kappa d})$ are nonempty and disjoint, and if both of these sets contain some benign agents, then under $(dF+1,F+1)$-robust graph, there exists at least one benign agent in either $\mathcal{V}^M(k,k+\kappa d,\epsilon_{\kappa d})$ or $\mathcal{V}^m(k,k+\kappa d,\epsilon_{\kappa d})$ that has no less than $dF+1$ neighboring agents outside its set. Suppose the benign agent $i\in \mathcal{V}^M(k,k+\kappa d,\epsilon_{\kappa d})$ is such a node. Following a similar proof procedure of Proposition \ref{thm:consensus_local}, we know that agent $i$ will always apply a state who has its first entry being no more than $M_1(k)-\epsilon_{\kappa d}$ for updating, under Algorithm \ref{alg:resilient}. The result can be finally concluded by applying the proof techniques as before.

\end{document}